\pdfoutput=1
\documentclass[acmsmall]{acmart}

\usepackage{amsfonts}
\usepackage{amsthm}
\usepackage{amsmath}
\usepackage{amscd}
\usepackage{indentfirst}
\usepackage{fancyvrb}
\usepackage{graphicx}
\usepackage{graphics}
\usepackage{mathpartir}
\numberwithin{equation}{section}
\usepackage[labelfont=bf]{caption}
\usepackage{pl-syntax}

\usepackage[switch]{lineno}
\usepackage{tikz-cd}
\usepackage{multicol}
\usepackage{xspace}

\usepackage{ stmaryrd }

\usepackage{mdframed}
\usepackage{xcolor}
\usepackage{tcolorbox}
\newcommand\makenote[4]{\par\noindent\begin{tcolorbox}
[size=small,colback=#2]\color{#3}\textbf{#4: }#1\end{tcolorbox}\noindent}
\definecolor{cornellred}{RGB}{196,18,48}
\definecolor{dartmouthgreen}{RGB}{0,112,60}
\definecolor{cisorange}{RGB}{247,147,33}
\definecolor{cisblue}{RGB}{67,199,244}
\definecolor{purple}{RGB}{96, 24, 168}

\newcommand\pa[1]{\makenote{#1}{cornellred}{white}{PA}}

\newcommand\orange[1]{\textcolor{orange}{#1}}
\newcommand\purple[1]{\textcolor{purple}{#1}}

\newcommand{\cat}{%
\mathbf%
}

\newcommand{\sem}[1]{
  \left\llbracket #1 \right\rrbracket
}

\newcommand{\bsem}[1]{
  \llparenthesis #1 \rrparenthesis
}

\newcommand{\D}{D}

\newcommand{\Pcal}{\mathcal{P}}

\usepackage{pgfplots}

\newcommand{\set}[2]{\{#1 \mid #2 \}}

\newcommand{\M}{\mathcal{M}}
\newcommand{\Rel}{\mathcal{R}}
\newcommand{\xto}{\xrightarrow}

\newcommand{\smid}{\ \mid \ }

\newcommand{\subst}[3]{#1\{#3 / #2\}}

\newcommand{\nat}{\mathbb{N}}
\newcommand{\bool}{\mathbb{B}}
\newcommand{\R}{\mathbb{R}}

\newcommand{\shr}{\orange{NI}\xspace}
\newcommand{\sep}{\purple{I}\xspace}
\newcommand{\vdashni}{\vdash_{\shr}\, }
\newcommand{\vdashi}{\vdash_{\sep}\, }

\newcommand{\lamb}[2]{\lambda #1\ldotp  #2}
\newcommand{\app}[2]{#1 \ #2}
\newcommand{\coin}{\mathsf{coin}}

\newcommand{\lto}{\multimap}
\newcommand{\sample}[3]{\mathsf{sample}\ #1 \ \mathsf{as } \ #2 \ \mathsf{in}\ #3}
\newcommand{\send}[3]{\mathsf{send}\ #1 \ \mathsf{as } \ #2 \ \mathsf{in}\ #3}
\renewcommand{\ifthenelse}[3]{\mathsf{if }\, #1\, \mathsf{ then }\, #2\, \mathsf{ else }\, #3}
\newcommand{\caseof}[3]{\mathsf{case}\, #1 \, \mathsf{of}\, (|\mathsf{in}_1 x \Rightarrow #2 \mid \mathsf{in}_2 x \Rightarrow #3)}
\newcommand{\kwtt}{\mathsf{tt}}
\newcommand{\kwff}{\mathsf{ff}}

\newcommand{\return}{\mathsf{return}\ }

\newcommand{\letin}[3]{\mathsf{let}\ #1 = #2 \ \mathsf{in}\ #3}

\newcommand{\onelang}{$\lambda_{\text{INI}}$\xspace}
\newcommand{\twolang}{$\lambda_{\text{INI}}^2$\xspace}

\acmJournal{PACMPL}
\acmVolume{1}
\acmNumber{CONF} 
\acmArticle{1}
\acmYear{2018}
\acmMonth{1}
\acmDOI{} 
\startPage{1}

\setcopyright{none}

\usepackage{booktabs}   
\usepackage{subcaption} 
\usepackage{cleveref}[capitalize]

\theoremstyle{plain}
\newtheorem{Th}{Theorem}[section]
\newtheorem{Lemma}[Th]{Lemma}
\newtheorem{Cor}[Th]{Corollary}

\theoremstyle{definition}
\newtheorem{Def}[Th]{Definition}

\newtheorem{Rem}[Th]{Remark}
\newtheorem{?}[Th]{Problem}
\newtheorem{Ex}[Th]{Example}

\crefname{Th}{theorem}{theorems}
\Crefname{Th}{Theorem}{Theorems}
\crefname{Cor}{corollary}{corollaries}
\Crefname{Cor}{Corollary}{Corollaries}
\crefname{Prop}{proposition}{propositions}
\Crefname{Prop}{Proposition}{Propositions}
\crefname{Def}{definition}{definitions}
\Crefname{Def}{Definition}{Definitions}
\crefname{Rem}{remark}{remarks}
\Crefname{Rem}{Remark}{Remarks}
\crefname{Ex}{example}{examples}
\Crefname{Ex}{Example}{Examples}

\begin{document}

\title[]{Separated and Shared Effects in Higher-Order Languages}         


\author{Pedro H. Azevedo de Amorim}
\affiliation{
  \institution{Cornell University}            
  \city{Ithaca}
  \state{NY}
  \country{USA}
}
\email{pamorim@cs.cornell.edu}          

\author{Justin Hsu}
\affiliation{
  \institution{Cornell University}            
  \city{Ithaca}
  \state{NY}
  \country{USA}
}
\email{justin@cs.cornell.edu}          

\begin{abstract}
  Effectful programs interact in ways that go beyond simple input-output, making
  compositional reasoning challenging. Existing work has shown that when such
  programs are ``separate'', i.e., when programs do not interfere with each
  other, it can be easier to reason about them. While reasoning about separated
  resources has been well-studied, there has been little work on reasoning about
  separated effects, especially for functional, higher-order programming
  languages.

  We propose two higher-order languages that can reason about sharing and
  separation in effectful programs. Our first language \onelang has a linear
  type system and probabilistic semantics, where the two product types capture
  independent and possibly-dependent pairs.  Our second language \twolang is a
  two-level, stratified language, inspired by Benton's linear-non-linear (LNL)
  calculus. We motivate this language with a probabilistic model, but we also
  provide a general categorical semantics and exhibit a range of concrete models
  beyond probabilistic programming. We prove soundness theorems for all of our
  languages; our general soundness theorem for our categorical models of
  \twolang uses a categorical gluing construction.
\end{abstract}



\maketitle

\section{Introduction}


A central challenge in the theory of programming languages is to come up with
sound and expressive reasoning principles for effectful programs. In contrast
with pure programs, where different programs can only affect each other at
clearly defined interfaces (e.g., the input or output from a functional call),
the interaction between effectful programs can be subtle and difficult to reason
about. To simplify formal analysis, it is highly useful to know when different
effectful computations are \emph{separate}, i.e., they do not interfere with
each other. For instance, in the presence of effects such as memory allocation
or probability, it is useful to know when pointers do not refer to the same
location, or when random quantities must be independent. 

\subsubsection*{Prior Work: Reasoning About Resource Separation}

While separated \emph{effects} have received relatively little attention in the
literature, there is a long line of work on reasoning about separation of
\emph{resources}~\citep{DBLP:conf/csl/OHearnRY01,DBLP:journals/tcs/PymOY04}. The
concept of resource is ubiquitous in Computer Science and usually manifests
itself when effectful programs interact with the external world. For example,
when programming with memory allocation, the heap is a kind of resource; when
programming with probabilistic sampling, randomness can be seen as a resource. 

In some cases, it is useful to ensure that computations access resources
separately. When programming with pointers, different pointers that \emph{alias}
refer to the same address, making it difficult to reason about updates to the
heap; requiring that programs do not alias can make formal verification more
modular and compositional. In the example of probabilistic effects, separation
of resources corresponds to probabilistic independence, while general joint
distributions can share resources. Just like for other notions of separation,
independence can simplify reasoning about programs. For instance, if two parts
of a program produce independent distributions, their joint distribution will
only depend on their individual probabilities---there are no unexpected
probabilistic interaction between the two parts. Independence can also be an
interesting property to verify; for instance, in cryptographic protocols, basic
security properties can be stated in terms of independence~\citep{barthe2019}.
Prior work has developed program logics that can about independence in the
context of a first-order, imperative language~\citep{barthe2019}. Unfortunately,
it is unclear how to capture independence in higher-order languages.

\subsubsection*{Our Work}
We aim to develop a higher-order language that can reason about shared and
separated effects in a variety of contexts. The closest work in this area is the
bunched calculus~\citep{DBLP:journals/jfp/OHearn03}, the Curry-Howard
correspondent of the logic of Bunched
Implications~\citep{DBLP:journals/bsl/OHearnP99}. While
\citet{DBLP:journals/jfp/OHearn03} gives a presheaf model for the language and
develops a concrete model for reasoning about memory-manipulating programs,
other concrete models are harder to come by.  Indeed, there are no known models
for the bunched calculus that can accommodate probability, or other common
monadic effects besides state.


Throughout this work we will use probabilistic effects as our guiding example.
We start by using a resource interpretation of probabilistic samples to
establish independence: if two computations use disjoint resources (i.e.,
probabilistic samples), then they produce independent random quantities.  Our
perspective yields two linear, higher-order languages that can reason about
probabilistic independence. Both languages have a product type constructor
$\otimes$ that enforces independence, in the sense that closed programs of type
$\nat \otimes \nat$ should be denoted by independent distributions.

Our first language \onelang is an linear $\lambda$-calculus with two product
types: the $\otimes$ type constructor enforces that the components of the pair
do not share any resources, while the $\times$ type constructor allows the
components to share resources.  Intuitively, $\otimes$ captures pairs of
independent values, while $\times$ captures pairs of general, possibly-dependent
values. We give a denotational semantics to \onelang and prove its soundness
theorem: the product $\otimes$ ensures probabilistic independence. 

While conceptually clean, \onelang has limited expressivity. For instance,
extending it with sum types breaks the soundness property. In order to mitigate
these issues, we define a richer, two-level language \twolang, where the two
product types of \onelang are restricted to different layers. Intuitively, one
layer allows computations that share randomness, while the other layer prevents
computations from sharing randomness. To enable the layers to interact, the
independent language has a modality that allows to soundly import programs
written in the shared language. This design is inspired by recent work by
\citet{azevedodeamorim2022sampling}, who proposed a two-level language to
combine the sampling and linear operator semantics of probabilistic programming
languages. We show that \twolang supports two different kinds of sum types: a
``shared'' sum in the sharing layer, and a ``separated'' sum in the independent
layer. We give a denotational semantics for the \twolang, prove soundness, and
give translations of two fragments of \onelang into \twolang.

\subsubsection*{Categorical Semantics and Concrete Models}

In order to show the generality of \twolang and how it connects to other classes
of effects, we propose a categorical semantics for \twolang and prove a general
soundness theorem of our type system. Then, we present concrete models of our
language inspired by a variety of existing effectful programming languages.

\begin{itemize}
  \item \textbf{Linear logic.}
    Models of linear logic have been used to give semantics to probabilistic
    languages~\citep{pcoh,stablecones,azevedodeamorim2022riesz}. We show that
    pairing these models with the category of Markov kernels yields models for
    \twolang. Our soundness theorem guarantees probabilistic independence; as
    far as we know, our method is the first to ensure independence in these
    models.
  \item \textbf{Distributed programming.}
    Next, we develop a relational model of \twolang for distributed programming.
    In this model, programs describe the implementation and communication
    patterns of multiple agents. Our soundness theorem shows that global
    programs of type $\tau_1 \otimes \tau_2$ can be compiled into two local
    programs that execute independently. This property is reminiscent of
    projection properties in choreographic languages \citep{choreographies}.
  \item \textbf{Name generation.}
    Programming languages with name generation include a primitive that
    generates a fresh identifier. In some contexts, it is important to control
    when and how many times a name is generated; for instance, reusing a
    \emph{nonce} value (``number once'') in cryptographic applications may make
    a protocol vulnerable to replay attacks. We define a model of \twolang based
    on name generation. Our soundness theorem states that the connective
    $\otimes$ enforces disjointness of the names used in each component.
  \item \textbf{Commutative effects.}
    We generalize the name generation and finite distribution models by noting
    that they are both example of monadic semantics of commutative effects.
    Under mild assumptions, every commutative monad gives rise to a model of
    \twolang.
  \item \textbf{Bunched and separation logics.}
    A long line of work uses \emph{bunched logics} to reason about separation of
    resources~\citep{DBLP:journals/bsl/OHearnP99,DBLP:conf/csl/OHearnRY01}.  We
    show that all models of affine bunched logics are also models of \twolang,
    but not vice-versa. To illustrate, we revisit O'Hearn's SCI+, a bunched type
    system for programming with memory
    allocation~\citep{DBLP:journals/jfp/OHearn03}. We define a model of \twolang
    based on SCI+, and give a sound translation of \twolang into SCI+.
\end{itemize}
The diversity of models suggests that \twolang is a suitable framework to reason
about separation and sharing in effectful higher-order programs.

\paragraph*{Outline.}
After reviewing mathematical preliminaries (\S \ref{sec:prelim}), we present our
main contributions:
\begin{itemize}
\item First, we define a linear, higher-order probabilistic $\lambda$-calculus called
  \onelang, with types that can capture probabilistic independence and
  dependence. We give a denotational semantics of our language and prove that
  $\otimes$ captures probabilistic independence (\S \ref{sec:langindep}). 
\item Next, we define a two-level, higher-order probabilistic $\lambda$-calculus
  called \twolang. This language combines an independent fragment and a sharing
  fragment with two distinct sum types: an independent sum, and a sharing sum.
  We give a probabilistic semantics and prove that $\otimes$ captures
  probabilistic independence; we also embed two fragments of \onelang into
  \twolang (\S \ref{sec:twolevel}).
\item Generalizing, we propose a categorical semantics for \twolang. Our
  semantics is a weaker version of Benton's linear/non-linear (LNL) model for
  linear logic~\citep{DBLP:conf/csl/Benton94} and of the calculus proposed
  by \citet{azevedodeamorim2022sampling} (\S \ref{sec:catmodel}).
\item We present a range of models for \twolang, described above. The soundness
  property of our type system ensures natural notions of independence in each of
  these models (\S \ref{sec:models}).
\item Finally, we prove a general soundness theorem: every program of type
  $\tau_1 \otimes \tau_2$ can be factored as two programs $t_1$ and $t_2$ of
  types $\tau_1$ and $\tau_2$, respectively. Our proof relies on a categorical
  gluing argument (\S \ref{sec:soundness}).
\end{itemize}
We survey related work in (\S \ref{sec:rw}), and conclude in (\S \ref{sec:conc}).

\section{Background}
\label{sec:prelim}

\subsection{Monads and their algebras}

We will assume knowledge of basic concepts from category theory, including
functors, products, coproducts, Cartesian closed categories, and symmetric
monoidal closed categories (SMCC). The interested reader can consult
\citet{leinster2014,mac2013} for good introductions to the subject.

\paragraph{Monads.}

Following seminal work by~\citet{DBLP:journals/iandc/Moggi91}, effectful
computations can be given a semantics via monads. A \emph{monad} over a category
$\cat{C}$ is a triple $(T, \mu, \eta)$ such that $T : \cat{C} \to \cat{C}$ is a
functor, $\mu_A : T^2 A \to T A$ and $\eta_A : A \to T A$ are natural
transformations such that $\mu_A \circ \mu_{TA} = \mu_A \circ T\mu_A$, $id_A =
\mu_A \circ T\eta_A$ and $id_A = \mu_A \circ \eta_{TA}$.

Another useful, and equivalent, definition of monads requires a natural
transformation $\eta_A$ and a lifting operation $(-)^* : \cat{C}(A, TB) \to
\cat{C}(TA, TB)$ such that objects from $\cat{C}$ and morphisms $A \to TB$ form
a category, usually referred to as the \emph{Kleisli category} $\cat{C}_T$. This
category has the same objects as $\cat{C}$, and has $Hom_{\cat{C_T}}(A, B) =
Hom_{\cat{C}}(A, TB)$. Kleisli categories are frequently used to give semantics
to effectful programming languages.


\paragraph*{Monad algebras.}
Given a monad $T$, a $T$-\emph{algebra} is a pair $(A, f : T A \to A)$ such that
$id_A = f \circ \eta_A$ and $f \circ \mu_A = f \circ Tf$. A $T$-\emph{algebra
morphism} $h : (A, f) \to (B, g)$ is a $\cat{C}$ morphism $h : A \to B$ such
that $g \circ Th = h \circ f$. $T$-algebras and morphisms form a
category $\cat{C}^T$, the \emph{Eilenberg-Moore category}.

\subsection{Probability Theory}

We will use probabilistic programs and effects to illustrate our higher-order
languages. 

\begin{Def}
    A distribution over a set $X$ is a function $\mu : X \to [0,1]$ such that
    $\sum_{x \in X} \mu(x) = 1$.
\end{Def}

Joint distributions are distributions over sets $X \times Y$. Given a joint distribution
$\mu$ over $X\times Y$, its marginal distribution over $X$ is defined as
$\mu_X(x) = \sum_{y \in Y}\mu(x, y)$ with and the second marginal $\mu_Y$ being similarly defined.

\begin{Def}
  A distribution $\mu$ over $X \times Y$ is probabilistically
  \emph{independent} if it is a product of its marginals $\mu_X$ and $\mu_Y$,
  i.e., $\mu(x,y) = \mu_X(x) \cdot \mu_Y(y)$, $x \in X$ and $y \in Y$.
\end{Def}

A probability monad can be defined for $\cat{Set}$. Given a set $X$, let 
$\D X$ be the set of functions $\mu : X \to [0,1]$ which are non-zero on 
finitely many values, and satisfy $\sum_{x \in supp(\mu)}\mu(x) = 1$~\citep{fritz2020}.
The unit of the monad is given by $\delta(a, b) = 1$ iff $a = b$ and $0$ otherwise,
while the bind is defined as $\mathsf{bind}(f)(\mu) = \sum_{x \in X} f(x)\mu(x)$.

\section{A Linear Language for Independence}
\label{sec:langindep}

To motivate our language for separated and shared effects, we will focus on one
effect: probabilistic sampling. We will build up two higher-order languages
where types can ensure probabilistic independence, the natural notion of
separation for probabilistic effects.

\subsection{Independence Through Linearity}
In many probabilistic programs, independent quantities are initially generated
through sampling instructions. Then, a simple way to reason about independence
of a pair of random expressions is to analyze which sources of randomness each
component uses: if the two expressions use distinct sources of randomness, then
they are independent; otherwise, they are possibly-dependent.

For instance, consider a simply typed first-order call-by-value language with a
primitive $\vdash \coin : \bool$ that flips a fair coin. The program
\[
  \letin x {\coin} {\letin y {\coin} {(x, y)}}
\]
flips two fair coins and pairs the results. This program will produce a
probabilistically independent distribution, since $x$ and $y$ are distinct
sources of randomness. On the other hand, the program 
\[
  \letin x {\coin} {(x, x)}
\]
does not produce an independent distribution: the two components are always
equal, and hence perfectly correlated. These principles are a natural fit for
substructural type systems, which control when variables can be shared. To
investigate this idea, we develop a language \onelang with an affine type system
that can reason about probabilistic independence.

\subsection{Introducing the Language \onelang}

\paragraph{Syntax.}
Figure~\ref{fig:abstract-syntax} presents the syntax of types and terms. Along
with base types ($\bool$), there are two product types: we view $\times$ as the
shared, or possibly-dependent product, while $\otimes$ is the independent
product. The language is higher-order, with a linear arrow type. The
corresponding term syntax is fairly standard. We have variables, numeric
constants, and primitive distributions ($\coin$). The two kinds of products can
be created from two kinds of pairs, and eliminated using projection and
let-binding, respectively. Finally, we have the usual $\lambda$-abstraction and
application. Our examples will use the standard syntactic sugar $\letin{x}{t}{u}
\triangleq \app{(\lamb x u)}{t}$.

\begin{figure}[t]
  \begin{syntax}
    \abstractCategory[Variables]{x, y, z}


    \category[Types]{\tau}
    \alternative{\bool}
    \alternative{\tau \times \tau}
    \alternative{\tau \otimes \tau}
    \alternative{\tau \lto \tau}



    \category[Expressions]{t, u}
    \alternative{x}
    \alternative{b \in \bool}
    \alternative{\coin}
    \alternative{(t, u)}
    \alternative{\pi_i \, t}
    \alternativeLine{t \otimes u}
    \alternative{\letin{x \otimes y}{t}{u}}
    \alternative{\lamb x t}
    \alternative{\app t u}
    \category[Contexts]{\Gamma}
    \alternative{x_1 : \tau_1, \ldots, x_n : \tau_n}
  \end{syntax}
  \hrule
  \caption{Types and Terms: \onelang}
  \label{fig:abstract-syntax}
\end{figure}

\paragraph{Type system.}

\Cref{fig:1type} shows the typing rules for \onelang; the rules are standard
from linear logic. The variable rule \textsc{Var} is \emph{affine}: variables 
in the context may not be used, and variables cannot be freely
duplicated. For the sharing product $\times$, the introduction rule
\textsc{$\times$ Intro} shares the context across the premises: both components
can use the same variables. Either component can be projected out of these pairs
(\textsc{$\times$ Elim$_i$}). For the independent product $\otimes$, in contrast,
the introduction rule \textsc{$\otimes$ Intro} requires both premises to use
\emph{disjoint} contexts. Thus, the components cannot share variables.  Tensor
pairs are eliminated by a let-pair construct that consumes both components
(\textsc{$\otimes$ Elim}). In substructural type systems, $\times$ is called an
\emph{additive} product, while $\otimes$ is called a \emph{multiplicative}
product. The abstraction and application rules are standard.

\paragraph{An additive arrow?}
Note that the application rule is multiplicative: the function cannot share
variables with its argument. A natural question is whether the arrow should be
additive: can we share variables between the function and its argument?
Substructural type systems like bunched
logic~\citep{DBLP:journals/bsl/OHearnP99} include both a multiplicative and an
additive arrow.

While we haven't defined the semantics of our language yet, we sketch an example
showing that an additive arrow would make it difficult for $\otimes$ to capture
probabilistic independence. If we allowed variables to be shared between the
function and its argument, we would be able to type-check:
\[
 \cdot \vdash \letin{x}{\coin}{\app{(\lamb{y}{x\otimes y})}{x}}: \bool \otimes \bool
\]
Under our semantics, which we will see next, this program is equivalent to
$\letin{x}{\coin{}}{x \otimes x}$, which produces a pair of correlated values.
Thus, we take a multiplicative arrow for our language.

\begin{figure}[t]
\begin{mathpar}
    \inferrule[Const]{~}{\cdot \vdash b : \bool}
    \and
    \inferrule[Coin]{~}{\cdot \vdash \coin : \bool}
    \and
    \inferrule[Var]{~}{\Gamma, x : \tau \vdash x : \tau}
    \and
    \inferrule[$\times$ Intro]{\Gamma \vdash t_1 : \tau \\ \Gamma \vdash t_2 : \tau_2}{\Gamma \vdash (t_1, t_2) : \tau_1 \times \tau_2}
    \and
    \inferrule[$\times$ Elim$_i$]{\Gamma \vdash t : \tau_1 \times \tau_2}{\Gamma \vdash \pi_i\, t : \tau_i}
    \\
    \inferrule[$\otimes$ Intro]{\Gamma_1 \vdash t_1 : \tau \\ \Gamma_2 \vdash t_2 : \tau_2}{\Gamma_1, \Gamma_2 \vdash t_1\otimes t_2 : \tau_1 \otimes \tau_2}
    \and
    \inferrule[$\otimes$ Elim]{\Gamma_1 \vdash t : \tau_1 \otimes \tau_2 \\ \Gamma_2, x : \tau_1, y : \tau_2 \vdash u : \tau}{\Gamma_1, \Gamma_2 \vdash \letin{x\otimes y}{t}{u} : \tau}
    \\
    \inferrule[Abstraction]{\Gamma, x : \tau_1 \vdash t : \tau_2}{\Gamma \vdash \lamb{x}{t} : \tau_1 \lto \tau_2}
    \and
    \inferrule[Application]{\Gamma_1 \vdash t : \tau_1 \lto \tau_2 \\ \Gamma_2 \vdash u : \tau_1}{\Gamma_1, \Gamma_2 \vdash \app{t}{u} : \tau_2}
\end{mathpar}
\hrule

\caption{Typing Rules: \onelang}
\label{fig:1type}

\end{figure}

\subsection{Denotational Semantics}

\begin{figure}
\begin{align*}
    \sem{\bool} &= \bool \\
    \sem{\tau \times \tau} &= \sem\tau \times \sem \tau \\
    \sem{\tau\otimes\tau} &= \sem\tau \times \sem\tau \\
    \sem{\tau_1 \lto \tau_2} &= \sem{\tau_1} \to \D\sem{\tau_2}\\
    \\
    \sem{x_1 : \tau_1, \dots, x_n : \tau_n} &= \sem{\tau_1} \times \cdots \times \sem{\tau_n}\\
    \\
    \sem{\Gamma \vdash t : \tau} &: \sem{\Gamma} \to \D\sem{\tau}\\
    \\
    \sem{x}(\gamma, v_x) &= \return v_x \\
    \sem{b}(*) &= \return b \\
    \sem{\coin}(*) &= \frac{1}{2}(\delta_{\kwtt} + \delta_{\kwff})\\
    \sem{(t_1, t_2)}(\gamma) &= x \leftarrow \sem{t_1}(\gamma ); y \leftarrow \sem{t_2}(\gamma ); \return (x, y)\\
    \sem{\pi_i\ t}(\gamma) &= (x, y) \leftarrow \sem{t}(\gamma); \return x\\
    \sem{t_1 \otimes t_2}(\gamma_1, \gamma_2) &= x \leftarrow \sem{t_1}(\gamma_1); y \leftarrow \sem{t_2}(\gamma_2); \return (x, y)\\
    \sem{\letin{x \otimes y}{t}{u}}(\gamma_1, \gamma_2) &= (x, y) \leftarrow \sem{t}(\gamma_1); \sem{u}(\gamma_2, x, y)\\
    \sem{\lamb x t}(\gamma) &= \return (\lamb x {\sem t}(\gamma))\\ 
    \sem{\app t u}(\gamma_1, \gamma_2) &= f \leftarrow \sem{t}(\gamma_1); x \leftarrow \sem {u}(\gamma_2); f(x)
\end{align*}
\hrule
\caption{Denotational Semantics: \onelang}
\label{fig:monsem}

\end{figure}

We can give a semantics to this language using the category $\cat{Set}$ and the
finite probability monad $\D$. From top to bottom, \Cref{fig:monsem} defines
the semantics of types, contexts, and typing derivations producing well-typed
terms. For types, we interpret both product types as products of sets. Arrow
types are interpreted as the set of Kleisli arrows, i.e., maps $\sem{\tau_1}
\to \D\sem{\tau_2}$. Contexts are interpreted as products of sets.

We interpret well-typed terms as Kleisli arrows. We briefly walk through the
term semantics, which is essentially the same as the Kleisli semantics proposed
by \citet{DBLP:journals/iandc/Moggi91}. Variables are interpreted using the unit
of the monad, which maps a value $v$ to the point mass distribution $\delta_v$.
Coins are interpreted as the fair convex combination of two point mass
distributions over $\kwtt$ and $\kwff$.

The rest of the constructs involve sampling, which is semantically modeled by
composition of Kleisli morphisms. We use monadic arrow notation to denote
Kleisli composition, i.e., $x \leftarrow f; g \triangleq g^* \circ f$. The two
pair constructors have the same semantics: we sample from each component, and
then pair the results. The projections for $\times$ computes the marginal of a
joint distribution, while let-binding for $\otimes$ samples from the pair $t$
and then uses the sample in the body $u$. Lambda abstractions are interpreted as
point mass distributions, while applications are interpreted as sampling the
function, sampling the argument, and then applying the first sample to the
second one.

\begin{Ex}[Correlated pairs]
It may seem as if there is no way of creating non-independent pairs, since the
semantics for both kinds of pairs samples each component independently. However,
consider the program $\letin{x}{\coin{}}{(x,x)}$. By unfolding the definitions,
its semantics is 
\[
x \leftarrow \frac{1}{2}(\delta_0 + \delta_1); y \leftarrow \delta_x; z \leftarrow \delta_x; \delta_{(y, z)}
= x \leftarrow \frac{1}{2}(\delta_0 + \delta_1); \delta_{(x, x)}
= \frac{1}{2}(\delta_{(0,0)} + \delta_{(1,1)}) .
\]
The resulting samples are perfectly correlated, not independent.
\end{Ex}

\begin{Ex}[Independent pairs are correlated pairs]
  Independent distributions are also possibly-dependent distributions. In
  \onelang, this fact is reflected by the following program:
  \[
    \cdot \vdash \lamb{z}{\letin{x \otimes y}{z}{(x, y)}} : \tau_1 \otimes \tau_2 \lto \tau_1 \times \tau_2 .
  \]
  If we unfold the semantics of this program, we see that this program does not
  modify the input.
\end{Ex} 

\subsection{Soundness}

The type system of \onelang guarantees that $\otimes$ enforces probabilistic
independence. Concretely, if $\cdot \vdash t : \tau_1 \otimes \tau_2$ is
well-typed, then $\sem{t}(*)$ is an independent probability distribution over
$\sem{\tau_1} \times \sem{\tau_2}$. We show this soundness theorem by
constructing a logical relation $\mathcal{R}_\tau \subseteq \D(\sem{\tau})$,
defined as:
\begin{align*}
   \Rel_{\bool} &= \D(\bool) \\
   \Rel_{\tau_1 \otimes \tau_2} &= \set{\mu_1 \otimes \mu_2 \in \D(\sem{\tau_1} \times \sem{\tau_2})}{\mu_i \in \Rel_{\tau_i}}\\
   \Rel_{\tau_1 \times \tau_2} &= \set{\mu \in \D(\sem{\tau_1} \times \sem{\tau_2})}{ \pi_i(\mu) \in \Rel_{\tau_i} \text{ for } i \in \{ 1, 2 \} }\\
   \Rel_{\tau_1 \lto \tau_2} &= \set{\mu \in \D(\sem{\tau_1} \to \D(\sem{\tau_2})) }{\forall \mu' \in \Rel_{\tau_1}, x \leftarrow \mu'; f \leftarrow \mu; f(x) \in R_{\tau_2}} .
\end{align*}

\begin{Th}
\label{th:soundkl}
  If $x_1 : \tau_1, \dots, x_n : \tau_n \vdash t : \tau$ and $\mu_i \in \Rel_{\tau_i}$ then
  \[
    (x_1 \leftarrow \mu_1; \cdots ; x_n \leftarrow \mu_n ;  \sem{t}(x_1, \dots, x_n)) \in \Rel_\tau .
  \]
\end{Th}
\begin{proof}
  Let the distribution above be $\nu$. We write $\overline{x_i}$ as shorthand
  for $x_1, \dots, x_n$, and $\overline{x_i \leftarrow \mu_i}$ as shorthand for
  $x_1 \leftarrow \mu_1; \cdots ; x_n \leftarrow \mu_n$. We prove $\nu \in
  \Rel_\tau$ by induction on the derivation of $\Gamma \vdash t : \tau$.
  \begin{description}
    \item[\textsc{Const}/\textsc{Coin}/\textsc{Var}.] Trivial.
      For instance, \textsc{Var}: $\nu = \overline{x_i \leftarrow \mu_i} ; \return
      x_i = \mu_i$ is in $\Rel_{\tau_i}$ by assumption.
    \item[$\times$ \textsc{Intro}.] We have $\nu = \overline{x_i \leftarrow
      \mu_i}; x \leftarrow \sem{t_1}(\overline{x_i}); y \leftarrow
      \sem{t_2}(\overline{x_i}); \return (x,y)$. It is straightforward to show
      that the first marginal of $\nu$ is $\overline{x_i \leftarrow \mu_i}; x
      \leftarrow \sem{t_1}(\overline{x_i}); \return x$ which, by the induction
      hypothesis, in an element of $\Rel_{\tau_1}$; similarly, the second
      marginal of $\nu$ is an element of $\Rel_{\tau_2}$.
    \item[$\times$ \textsc{Elim}.] We have $\nu = \overline{x_i \leftarrow
      \mu_i}; (x, y) \leftarrow \sem{t}(\overline{x_i}); \return x$. By the
      induction hypothesis, $\sem{t}(x_i) \in \Rel_{\tau_1 \times \tau_2}$ and,
      by assumption, its marginals are elements of $\Rel_{\tau_1}$ and
      $\Rel_{\tau_2}$.
    \item[$\otimes$ \textsc{Intro}.] Let $\overline{\mu_i}$ be the sequence of
      distributions corresponding to $\Gamma_1$, and let $\overline{\eta_i}$ be
      the sequence of distributions corresponding to $\Gamma_2$. Since $\D$ is a
      commutative monad~\citep{borceux1994}, we may apply associativity and
      commutativity to show:
      \begin{align*}
        \nu &= x_i \leftarrow \overline{\mu_i};
        y_i \leftarrow \overline{\eta_i};
        x \leftarrow \sem{t_1}(\overline{x_i});
        y \leftarrow \sem{t_2}(\overline{y_i});
        \return (x, y) \\
            &= \overline{x_i \leftarrow \mu_i};
            x \leftarrow \sem{t_1}(\overline{x_i});
            \overline{y_i \leftarrow \eta_i};
            y \leftarrow \sem{t_2}(\overline{y_i});
            \return (x,y) \\
            &= (\overline{x_i \leftarrow \mu_i};
            x \leftarrow \sem{t_1}(\overline{x_i});
            \return x)
            \otimes
            (\overline{y_i \leftarrow \eta_i};
            y \leftarrow \sem{t_2}(\overline{y_i});
            \return y)
            = \nu_1 \otimes \nu_2 .
      \end{align*}
      Furthermore, by induction hypothesis, $\nu_i \in \Rel_{\tau_i}$ so $\nu =
      \nu_1 \otimes \nu_2 \in \Rel_{\tau_1 \otimes \tau_2}$ as desired.
    \item[$\otimes$ \textsc{Elim}.] Let $\overline{\mu_i}$ be the sequence of
      distributions corresponding to $\Gamma_1$, and let $\overline{\eta_i}$ be
      the sequence of distributions corresponding to $\Gamma_2$. We have:
      \begin{align*}
        \nu &= \overline{x_i \leftarrow \mu_i};
        \overline{y_i \leftarrow \eta_i};
        (x, y) \leftarrow \sem{t}(\overline{x_i}); \\
            &= \overline{x_i \leftarrow \mu_i};
            (x, y) \leftarrow \sem{t}(\overline{x_i});
            \overline{y_i \leftarrow \eta_i};
            \sem{u}(\overline{y_i}, x, y) \\
            &= (x, y) \leftarrow \nu_1 \otimes \nu_2;
            \overline{y_i \leftarrow \eta_i};
            \sem{u}(\overline{y_i}, x, y) \\
            &= \overline{y_i \leftarrow \eta_i};
            x \leftarrow \nu_1;
            y \leftarrow \nu_2;
            \sem{u}(\overline{y_i}, x, y)
      \end{align*}
      where the third equality is by the induction hypothesis from the first
      premise. By the induction hypothesis from the second premise, the final
      distribution is in $\Rel_{\tau}$, as desired.
    \item[\textsc{Abstraction}.] By unfolding the definitions, we need to show
      \[
        x \leftarrow \mu; f \leftarrow (x_i \leftarrow \mu_i; \delta_{\lamb x
        {\sem t(x_i)}}); f(x) \in \Rel_{\tau_2} ,
      \]
      for some $\mu \in \Rel_{\tau_1}$. This distribution is equal to $x_i
      \leftarrow \mu_i; x \leftarrow \mu; f \leftarrow \delta_{\lamb x {\sem
      t(x_i)}}; f(x)$, by associativity and commutativity. By the induction
      hypothesis and the fact that $\delta$ is the unit of the monad, we can
      conclude this case.
    \item[\textsc{Application}.] This case follows directly from the induction
      hypotheses. \qedhere
  \end{description}
\end{proof}

Our soundness property for \onelang follows immediately.

\begin{Cor}
  If $\cdot \vdash t : \tau_1 \otimes \tau_2$ then $\sem{t}(*)$ is an
  independent probability distribution over $\sem{\tau_1} \times \sem{\tau_2}$.
\end{Cor}

\section{A Two-Level Language for Independence}
\label{sec:twolevel}

The affine type system of \onelang can distinguish between independent and
possibly dependent random quantities, but the language is not as expressive as
we would like. We first discuss these limitations, and then introduce a
stratified, two-level language \twolang that resolves these problems. Finally,
we show how to embed two fragments of \onelang into \twolang.

\subsection{Limitations of \onelang: Sums and Let-Bindings}

\paragraph*{Adding sum types.}
Though there are base types like $\bool$ in \onelang, there are no conditionals.
Extending \onelang with sum types and case analysis immediately leads to
problems. Consider the program:
\[
  \ifthenelse{\coin}{\kwtt \otimes \kwtt}{\kwff \otimes \kwff}
\]
Operationally, this probabilistic program flips a fair coin and a pair with two
copies of the result, $\kwtt \otimes \kwtt$ or $\kwff \otimes \kwff$.  Since
$\kwtt$ and $\kwff$ are constants they do not share any variables, so both
branches can be given type $\bool \otimes \bool$ and a standard case analysis
rule would assign the whole program $\bool \otimes \bool$. However, this
extension would break soundness (\cref{th:soundkl}): the pair is not
probabilistically independent because its components are always equal to each
other.

This example illustrates that we should not allow case analysis to produce
programs of type $\tau_1 \otimes \tau_2$. However, note that it is safe to allow
case analysis to produce programs of type $\tau_1 \times \tau_2$ since this
product does not assert independence. Thus, incorporating sum types into
\onelang while preserving soundness seems to require ad hoc restrictions on the
elimination rule.

\paragraph*{Reusing variables.}
Another restriction is that function application is multiplicative.  The
limitation can be seen when using let-bindings, which are syntactic sugar for
application. In $\letin{x}{t}{u}$, the terms $t$ and $u$ \emph{cannot} share any
variables. For instance, \onelang does not allow the following program:
\begin{align*}
  &\letin{x_1}{\coin}{ }
  \letin{x_2}{\coin}{ }\\
  &\qquad\letin{y}{f(x_1, x_2)}{ }
  \letin{z}{g(x_1, x_2)}{(y, z)}
\end{align*}
However, there are useful sampling algorithms (e.g., the Box-Muller
transform~\citep{boxmuller}) that follow this shape. In order to write a
well-typed version of this program in \onelang, we could inline the definitions
of $y$ and $z$: the pair constructor $(-, -)$ is additive, so the two components
can both use $x_1$ and $x_2$. However, it is awkward to require this change.

Similarly, given a term of type $\tau_1 \times \tau_2$, we can't directly
project out both components at the same time. For instance, the program
\[
  \letin{x}{\pi_1\, z}{%
    \letin{y}{\pi_2\, z}{%
      f(x, y)
  }}
\]
is not well-typed, since the outer let-binding shares the variable $z$ with its
body. These problems would be solved if function application in \onelang was
additive; however, as we saw in \Cref{sec:langindep}, allowing a function and an
argument to share variables can also break soundness of \onelang.

\subsection{The Language \twolang: Syntax, Typing Rules and Semantics}

To address these limitations, we introduce a stratified language. We are guided
by a simple observation about products, sums, and distributions, which might be
of more general interest. In \onelang, the product types correspond to two
distinct ways of composing distributions with products: the sharing product
$\tau_1 \times \tau_2$ corresponds to \emph{distributions of products},
$M(\tau_1 \times \tau_2)$, while the separating product $\tau_1 \otimes \tau_2$
corresponds to \emph{products of distributions}, $M\tau_1 \times M\tau_2$.

Similarly, there are two ways of combining distributions and sums:
\emph{distributions of sums}, $M(\tau_1 + \tau_2)$, and \emph{sums of
distributions}, $M\tau_1 + M\tau_2$. We think of the first combination as a
\emph{sharing sum}, since the distribution can place mass on both components of
the sum. In contrast, the second combination is a \emph{separating sum}, since
the distribution either places all mass on $\tau_1$ or all mass on $\tau_2$.

Finally, there are interesting interactions between sharing and separating, sums
and products. For instance, the problematic sum example we saw above performs
case analysis on $\coin$---a sharing sum, because it has some probability of
returning true and some probability of returning false---but produces a
separating product $\bool \otimes \bool$. If we instead perform case analysis on
a \emph{separating} sum, then the program either always takes the first branch
or always takes the second branch, and now there is no problem with producing a
separating product.

These observations lead us to design a two-level language, where one layer
includes the sharing connectives and the other layer includes the separating
connectives. We call this language \twolang, where INI stands for
\emph{independent/non-independent}.


\paragraph*{Syntax.} 
The program and type syntax of \twolang, summarized in Figure~\ref{fig:emlang},
is stratified into two layers: a non-independent (\shr) layer, and an
independent (\sep) layer. We will color-code them: the \shr-language will be
\orange{orange}, while the \sep-language will be \purple{purple}.

The \shr layer has base, product ($\textcolor{orange}{\times}$), and sum types
($\textcolor{orange}{+}$). The language is mostly standard: we have variables,
constants, basic distributions (\orange{$\coin$}), and a set
$\orange{\mathcal{O}(\tau_1, \tau_2)}$ of primitive operations from
$\orange{\tau_1}$ to $\orange{\tau_2}$, along with the usual pairing and
projection constructs for products, and injection and case analysis constructs
for sums. The \shr layer does not have arrows, but it does allow let-binding.

The \sep-layer is quite similar to \onelang: it has its own product
($\textcolor{purple}{\otimes}$) and sum ($\textcolor{purple}{\oplus}$) types,
and a linear arrow type ($\textcolor{purple}{\lto}$). The type
$\M(\orange{\tau})$ brings a type from the \shr-layer into the \sep-layer. The
language is also fairly standard, with constructs for introducing and
eliminating products and sums, and functions and applications. The last
construct
$(\sample{\textcolor{purple}{\overline{t}}}{\textcolor{orange}{\overline{x}}}{\textcolor{orange}{M}})$
is from \cite{azevedodeamorim2022sampling}: it allows the two layers to interact. Here, $\overline{t}$ and
$\overline{x}$ are two (possibly empty) lists of the same length.

Intuitively, the \shr-language allows sharing while the \sep-language disallows
sharing. Each language has its own sum type, a sharing and separated sum,
respectively, each of which interacts nicely with its own product type. The $\M$
modality can be thought of as an abstraction barrier between both languages that
enables the manipulation of shared programs in a separating program while not
allowing its sharing to be inspected, except when producing another boxed term.

\begin{figure}[h]
  \begin{syntax}
    \abstractCategory[Variables]{x, y, z}


    \category[\shr-types]{\textcolor{orange}{\tau}}
    \alternative{\textcolor{orange}{\bool}}
    \alternative{\textcolor{orange}{\tau \times \tau}}
    \alternative{\textcolor{orange}{\tau + \tau}}
    \category[\sep-types]{\textcolor{purple}{\underline\tau}}
    \alternative{\textcolor{purple}{\underline\tau \otimes \underline \tau}}
    \alternative{\textcolor{purple}{\underline\tau \oplus \underline \tau}}
    \alternative{\textcolor{purple}{\underline\tau \lto \underline \tau}}
    \alternative{\M(\textcolor{orange}{\tau})}


    \separate

    \category[\shr-expressions]{\textcolor{orange}{M, N}}
    \alternative{\textcolor{orange}{x}}
    \alternative{\textcolor{orange}{b} \in \bool}
    \alternative{\textcolor{orange}{\coin}}
    \alternative{\textcolor{orange}{f \in \mathcal{O}(\tau_1, \tau_2)}}
    \alternative{\textcolor{orange}{(M, N)}}
    \alternative{\textcolor{orange}{\pi_i \, M}}
    \alternative{\textcolor{orange}{\mathsf{in_i\ t}}}
    \alternativeLine{\textcolor{orange}{\caseof{t}{u_1}{u_2}}}
    \alternative{\textcolor{orange}{\letin{x}{M}{N}}}
    
    \category[\sep-expressions]{\textcolor{purple}{t, u}}
    \alternative{\textcolor{purple}{x}}
    \alternative{\textcolor{purple}{t \otimes u}}
    \alternative{\textcolor{purple}{\letin{x \otimes y}{t}{u}}}
    \alternative{\purple{\mathsf{op}}}
    \alternative{\textcolor{purple}{\mathsf{in_i\ t}}}
    \alternativeLine{\textcolor{purple}{\caseof{t}{u_1}{u_2}}}
    \alternative{\textcolor{purple}{\lamb x t}}
    \alternative{\textcolor{purple}{\app t u}}
    \alternative{\sample{\textcolor{purple}{\overline{t}}}{\textcolor{orange}{\overline{x}}}{\textcolor{orange}{M}}}

    \category[\shr-contexts]{\textcolor{orange}{\Gamma}}
    \alternative{\textcolor{orange}{x}_1 : \textcolor{orange}{\tau}_1, \ldots, \textcolor{orange}{x}_n : \textcolor{orange}{\tau}_n}

    \category[\sep-contexts]{\textcolor{purple}{\Gamma}}
    \alternative{\textcolor{purple}{x}_1 : \textcolor{purple}{\underline\tau}_1, \ldots, \textcolor{purple}{x}_n : \textcolor{purple}{\underline\tau}_n}
  \end{syntax}
\hrule
  \caption{Types and Terms: \twolang}
  \label{fig:emlang}
\end{figure}

\paragraph*{Typing rules.}
The typing rules of \twolang are presented in \Cref{fig:EMtyping}. We have two
typing judgments for the two layers; we use subscripts on the turnstiles to
indicate the layer. We start with the first group of typing rules, for the
sharing (\shr) layer. These typing rules are entirely standard for a first-order
language with products and sums. Note that all rules allow the context to be
shared between different premises. In particular, the let-binding rule is
\emph{additive} instead of multiplicative as in \onelang: a let-binding is
allowed to share variables with its body.

The second group of typing rules assigns types to the independent (\sep) layer.
These rules are the standard rules for multiplicative additive linear logic
(MALL), and are almost identical to the typing rules for \onelang. Just like
before, the rules treat variables affinely, and do not allow sharing variables
between different premises. The rules for the sum $\textcolor{purple}{\tau_1
\oplus \tau_2}$ are new. Again, the elimination (\textsc{Case}) rule does not
allow sharing variables between the guard and the body.

The final rule, \textsc{Sample}, gives the interaction rule between the two
languages. The first premise is from the sharing (\shr) language, where the
program $\textcolor{orange}{M}$ can have free variables $\textcolor{orange}{x_1,
\dots, x_n}$. The rest of the premises are from the independent (\sep) language,
where linear programs $\textcolor{purple}{t_i}$ have boxed type $\M
\orange{\tau_i}$. The conclusion of the rule combines programs
$\textcolor{purple}{t_i}$ with $\textcolor{orange}{M}$, producing an
\sep-program of boxed type. Intuitively, this rule allows a program in the
sharing language to be imported into the linear language.  Operationally,
$\sample {\textcolor{purple}{t}} {\textcolor{orange}{x}}
{\textcolor{orange}{M}}$ constructs a distribution $\textcolor{purple}{t}$ using
the independent language, samples from it and binds the sample to
$\textcolor{orange}{x}$ in the shared program $\textcolor{orange}{M}$, and
finally boxes the result into the linear language.


\paragraph*{Probabilistic Semantics}
To keep the presentation concrete, in this section we will work with a concrete
semantics motivated by probabilistic independence, where programs are
probabilistic programs with discrete sampling. In the next section, we will
present the general categorical semantics of \twolang and consider other models.

The probabilistic semantics for \twolang is defined in Figure~\ref{fig:EMsem}.
For the \shr-layer, we use the same semantics of \onelang, i.e., well-typed
programs are interpreted as Kleisli arrows for the finite distribution monad
$\D$. The Kleisli category $\cat{Set}_\D$ has sets as objects, so we may simply
define the semantics of each type to be a set. It is also known that
$\cat{Set}_\D$ has products and coproducts, which can be used to interpret
well-typed programs in \shr.

For the $\sep$-language, we use the category of algebras for the finite
distribution monad $\D$ and plain maps, $\widetilde{\cat{Set}^\D}$.  Concretely,
its objects are pairs $(A, f)$, where $f$ is an $M$-algebra, and a morphism $(A,
f) \to (B, g)$ is a function $A \to B$. Given two objects $(A, f)$ and $(B, g)$
we can define a product algebra over the set $A \times B$. Furthermore, it is also possible
to equip the set-theoretic disjoint union $A + B$ and exponential $A \Rightarrow
B$ with algebra structures, making it a model of higher-order programming with
case analysis \cite{simpson1992}. We only need to explicitly define the algebraic structure when
interpreting the type constructor $\M$, which is interpreted as the free
$\D$-algebra with the multiplication for the monad as the algebraic structure.



\begin{figure}
\begin{mathpar}
  \inferrule[Const]{ \orange{b} \in \orange{\bool}}{\textcolor{orange}{\Gamma} \vdashni \orange b : \textcolor{orange}{\bool}}
\and
  \inferrule[Primitive]{\textcolor{orange}\Gamma \vdashni \textcolor{orange}M : \textcolor{orange}{\tau_1} \\ \textcolor{orange}f \in \orange{\mathcal{O}_{NI}(\tau_1, \tau_2)}}{\textcolor{orange}\Gamma \vdashni \textcolor{orange}{f(M)} : \textcolor{orange}{\tau_2}}
\\
  \inferrule[Var]{~}{\textcolor{orange}{\Gamma}, \textcolor{orange}{x} : \textcolor{orange}{\tau} \vdashni \textcolor{orange}{x} : \textcolor{orange}{\tau}}
\and
  \inferrule[Let]{\textcolor{orange}{\Gamma} \vdashni \textcolor{orange}{t} : \textcolor{orange}{\tau_1} \\ \textcolor{orange}{\Gamma} , \textcolor{orange}{x} : \textcolor{orange}{\tau_1} \vdashni \textcolor{orange}{u} : \textcolor{orange}{\tau}}{\textcolor{orange}{\Gamma} \vdashni \textcolor{orange}{\letin x t u} : \textcolor{orange}{\tau}}
\\
  \inferrule[$\times$ Intro]{\textcolor{orange}{\Gamma} \vdashni \textcolor{orange}{M} : \textcolor{orange}{\tau_1} \\ \textcolor{orange}{\Gamma} \vdashni \textcolor{orange}{N} : \textcolor{orange}{\tau_2}}{\textcolor{orange}{\Gamma} \vdashni \textcolor{orange}{(M, N)} : \textcolor{orange}{\tau_1 \times \tau_2}}
\and
  \inferrule[$\times$ Elim$_i$]{\textcolor{orange}{\Gamma} \vdashni \textcolor{orange}{M} : \textcolor{orange}{\tau_1 \times \tau_2}}{ \textcolor{orange}{\Gamma} \vdashni \textcolor{orange}{\pi_i M} : \textcolor{orange}{\tau_i}}
\\
  \inferrule[$\oplus$ Intro$_i$]{ \textcolor{orange}{\Gamma} \vdashni \textcolor{orange}{M} : \textcolor{orange}{\tau_i}}{\textcolor{orange}{\Gamma} \vdashni \textcolor{orange}{\mathsf{in}_i \, M} : \textcolor{orange}{\tau_1 + \tau_2}}
\and
  \inferrule[$\oplus$ Elim]{ \textcolor{orange}{\Gamma} \vdashni \textcolor{orange}{M} : \textcolor{orange}{\tau_1 + \tau_2} \\ \textcolor{orange}{\Gamma, x : \tau_1} \vdashni \textcolor{orange}{N_1} : \textcolor{orange}{\tau} \\ \textcolor{orange}{\Gamma, x : \tau_2} \vdashni \textcolor{orange}{N_2} : \textcolor{orange}{\tau} }{\textcolor{orange}{\Gamma} \vdashni \textcolor{orange}{\mathsf{case}\, M \, \mathsf{of}\, (|\, \mathsf{in_1}\, x \Rightarrow N_1 \, |\, \mathsf{in_2}\, y \Rightarrow N_2)} : \textcolor{orange}{\tau}}
\\
\\
  \inferrule[Var]{ }{\textcolor{purple}{\Gamma}, \textcolor{purple}{x} : \purple {\underline{\tau}} \vdashi \textcolor{purple}{x} : \purple {\underline{\tau}}}
\and
  \inferrule[Operations] {\purple{\mathsf{op}} \in \purple{\mathcal{O}_{I}(\underline{\tau_1}, \underline{\tau_2})}} {\purple \Gamma \vdashi \purple{\mathsf{op}} : \purple{\underline{\tau_1}} \lto \purple{\underline{\tau_2}}}
\\
  \inferrule[Abstraction]{\textcolor{purple}{\Gamma}, \textcolor{purple}{x} : \purple {\underline{\tau_1}} \vdashi \textcolor{purple}{t} : \purple {\underline{\tau_2}}}{\textcolor{purple}{\Gamma} \vdashi \textcolor{purple}{\lamb x t} : \textcolor{purple}{\underline{\tau_1} \lto \underline{\tau_2}}}
\and
  \inferrule[Application]{\textcolor{purple}{\Gamma_1} \vdashi \textcolor{purple}{t} : \textcolor{purple}{\underline{\tau_1} \lto \underline{\tau_2}} \\ \textcolor{purple}{\Gamma_2} \vdashi \textcolor{purple}{u} : \textcolor{purple}{\underline{\tau_1}}}{\textcolor{purple}{\Gamma_1}, \textcolor{purple}{\Gamma_2} \vdashi \textcolor{purple}{\app t u} : \textcolor{purple}{\underline{\tau_2}}}
\\
  \inferrule[$\otimes$ Intro]{\textcolor{purple}{\Gamma_1} \vdashi \textcolor{purple}{t} : \textcolor{purple}{\underline{\tau_1}} \\ \textcolor{purple}{\Gamma_2} \vdashi \textcolor{purple}{u} : \textcolor{purple}{\underline{\tau_2}}}{\textcolor{purple}{\Gamma_1}, \textcolor{purple}{\Gamma_2} \vdashi \textcolor{purple}{t \otimes u} : \textcolor{purple}{\underline{\tau_1} \otimes \underline{\tau_2}}}
\and
  \inferrule[$\otimes$ Elim]{\textcolor{purple}{\Gamma_1}  \vdashi \textcolor{purple}{t} : \textcolor{purple}{\underline{\tau_1} \otimes \underline{\tau_2}} \\ \textcolor{purple}{\Gamma_2}, \textcolor{purple}{x} : \textcolor{purple}{\underline{\tau_1}}, \textcolor{purple}{y} : \textcolor{purple}{\underline{\tau_2}} \vdashi \textcolor{purple}{u} : \textcolor{purple}{\underline{\tau}}}{\textcolor{purple}{\Gamma_1}, \textcolor{purple}{ \Gamma_2} \vdashi \textcolor{purple}{\letin {x\otimes y} t u} : \textcolor{purple}{\underline{\tau}}}
\\
  \inferrule[$\oplus$ Intro$_i$]{ \textcolor{purple}{\Gamma} \vdashi \textcolor{purple}{t} : \textcolor{purple}{\underline{\tau_i}}}{\textcolor{purple}{\Gamma} \vdashi \textcolor{purple}{\mathsf{in}_i \, t} : \textcolor{purple}{\underline{\tau_1} \oplus \underline{\tau_2}}}
\and
  \inferrule[$\oplus$ Elim]{ \textcolor{purple}{\Gamma_1} \vdashi \textcolor{purple}{t} : \textcolor{purple}{\underline{\tau_1} \oplus \underline{\tau_2}} \\ \textcolor{purple}{\Gamma_2, x : \underline{\tau_1}} \vdashi \textcolor{purple}{u_1} : \textcolor{purple}{\underline{\tau}} \\ \textcolor{purple}{\Gamma_2, y : \underline{\tau_2}} \vdashi \textcolor{purple}{u_2} : \textcolor{purple}{\underline{\tau}} }{\textcolor{purple}{\Gamma_1, \Gamma_2} \vdashi \textcolor{purple}{\mathsf{case}\, t \, \mathsf{of}\, (|\, \mathsf{in_1}\, x \Rightarrow u_1 \, |\, \mathsf{in_2}\, y \Rightarrow u_2)} : \textcolor{purple}{\underline{\tau}}}
\\
\\
  \inferrule[Sample]{\textcolor{orange}{x_1} : \textcolor{orange}{\tau_1}, \dots
  , \textcolor{orange}{x_n} : \textcolor{orange}{\tau_n} \vdashni
\textcolor{orange}{M} : \textcolor{orange}{\tau} \quad
\textcolor{purple}{\Gamma_i} \vdashi \textcolor{purple}{t_i} :
\M(\textcolor{orange}{\tau_i}) \\ 0 < i \leq n}{\textcolor{purple}{\Gamma_1},
\dots, \textcolor{purple}{\Gamma_n} \vdashi \sample {\textcolor{purple}{t_i}}
{\textcolor{orange}{x_i}} {\textcolor{orange}{M}} : \M (\textcolor{orange}{\tau})}
\end{mathpar}
\hrule
\caption{Typing Rules: \twolang}
\label{fig:EMtyping}

\end{figure}

\begin{figure}
  \begin{align*}
    \bsem{\bool} &= \bool
                 &\sem{\M \tau} &= (\D\sem{\tau}, \mu_{\sem{\tau}})
    \\
    \bsem{\tau \times \tau} &= \bsem\tau \times \bsem \tau
                            &\sem{\underline\tau \otimes \underline\tau} &= \sem{\underline\tau} \times \sem{\underline\tau}
    \\
    \bsem{\tau + \tau} &= \bsem\tau + \bsem \tau
                       &\sem{\underline\tau \oplus \underline\tau} &= \sem{\underline\tau} + \sem{\underline\tau}
    \\
                       && \sem{\underline\tau \lto \underline\tau} &= \sem{\underline\tau} \to \sem{\underline\tau}
    \\
    \bsem{x_1 : \tau_1, \dots, x_n : \tau_n} &= \bsem{\tau_1} \times \cdots \times \bsem{\tau_n}
                                             &\sem{x_1 : \underline{\tau}_1, \dots, x_n : \underline{\tau}_n} &= \sem{\underline{\tau}_1} \times \cdots \times \sem{\underline{\tau}_n}
    \\
    \bsem{\Gamma \vdash M : \tau} &\in \cat{Set}_\D(\bsem{\Gamma}, \bsem{\tau})
                                  &\sem{\Gamma \vdash t : \underline{\tau}} &\in \widetilde{\cat{Set}^\D}(\sem{\Gamma}, \sem{\underline{\tau}})
  \end{align*}
  \hrule
  \begin{align*}
    \sem{x}(\gamma, v_x) &= v_x \\
    \sem{t \otimes u}(\gamma_1, \gamma_2) &= \sem{t}(\gamma_1) \times \sem{u}(\gamma_2)\\
    \sem{\letin{x \otimes y}{t}{u}}(\gamma_1, \gamma_2) &= \sem{u}(\gamma_2, \sem{t}(\gamma_1))\\
    \sem{\lamb x t}(\gamma)(x) &= \sem{t}(\gamma)(x) \\ 
    \sem{\app t u}(\gamma_1, \gamma_2) &= \sem{t}(\gamma_1, \sem{u}(\gamma_2)\\ 
    \sem{\mathsf{in}_i t}(\gamma) &= in_i(\sem{t}(\gamma))\\
    \sem{\caseof{t}{u_1}{u_2}}(\gamma_1, \gamma_2) &= \begin{cases}
      \sem{u_1}(\gamma_2, v), & \sem{t}(\gamma_1) = in_1(v) \\
      \sem{u_2}(\gamma_2, v), & \sem{t}(\gamma_1) = in_2(v)
    \end{cases}\\
      \sem{\sample{t_i}{x_i}{N}} &= \mu \circ \D\bsem{N} \circ (\sem{t_1} \times \cdots \times \sem{t_n})
    \end{align*}

  \hrule

\caption{Concrete Semantics: \twolang}
\label{fig:EMsem}

\end{figure}

Now that we have defined the probabilistic semantics of the \twolang, we can
prove its soundness theorem: just like in \onelang, the type constructor
$\otimes$ enforces probabilistic independence.

\begin{Th}
\label{th:soundem}
  If $\cdot \vdashi t : \M\tau_1 \otimes \M\tau_2$ then $\sem{t}$ is an independent distribution.
\end{Th}
\begin{proof}
  The semantics of $\cdot \vdashi t : \M\tau_1 \otimes \M\tau_2$ is a
  set-theoretic function $\sem{t} : 1 \to \D \sem{\tau_1} \times \D
  \sem{\tau_2}$, which is isomorphic to an independent distribution.
\end{proof}






\subsection{Revisiting Sums and Let-Binding}




Let us revisit the problematic if-then-else program at the beginning of the
section. The type system of \twolang makes it impossible to produce an
independent pair by pattern matching on values:
\[
\mathsf{dist} : \M(1 + 1) \nvdash_{\sep}\, \ifthenelse{\mathsf{dist}}{(\kwtt \otimes \kwtt)}{(\kwff \otimes \kwff)} : \M \bool \otimes \M \bool
\]
where if-statements are simply elimination of sum types over booleans. However,
we can write a well-typed version of this program if we use the sharing product:
\[
    \mathsf{dist} : \M(1 + 1) \vdashi \sample{\mathsf{dist}}{x}{(\ifthenelse{x}{(\kwtt, \kwtt)}{(\kwff, \kwff)})} : \M(\bool \times \bool)
\]
The design of \twolang also removes the limitations on let-bindings we discussed
before, since the sharing layer has an \emph{additive} let-binding. In
particular, it is also possible to express the problematic let-binding program
we saw before:
  \begin{align*}
    \cdot \vdashi &\sample{\coin, \coin}{x_1, x_2}{ }\\
                 &\qquad\letin{y}{f(x_1, x_2)}{ }
                 \letin{z}{g(x_1, x_2)}{ }
                 M : \M(\tau)
  \end{align*}
We can also project both components out of pairs in the sharing layer:
  \[
    \cdot \vdashni \letin{x}{\pi_1\, M_1}{%
      \letin{y}{\pi_2\, M_2}{%
        M : \tau
    }}
  \]

\subsection{Embedding from \onelang to \twolang}

Now that we have seen both \onelang and \twolang, a natural question is how
these languages are related. We first show how to embed the fragment of \onelang
without arrow types into \twolang. The idea is that the semantics of \onelang is
given by a Kleisli category, so there is a translation into the \shr-layer of
\twolang. The types are translated as follows:
\begin{mathpar}
    \mathcal{T}(\bool) \triangleq \bool
    \and
    \mathcal{T}(\tau_1 \times \tau_2) = \mathcal{T}(\tau_1 \otimes \tau_2) \triangleq \mathcal{T}(\tau_1) \times \mathcal{T}(\tau_2)
\end{mathpar}
At the term-level, the translation is the identity function.

\begin{Th}
  If $\Gamma \vdash M : \tau$ in \onelang then $\mathcal{T}(\Gamma) \vdashni
  \mathcal{T}(M) : \mathcal{T}(\tau)$ in \twolang.
\end{Th}

Furthermore, this translation preserves equations between programs and is fully abstract.

\begin{Th}
  Let $\Gamma \vdash t_1 : \tau$ and $\Gamma \vdash t_2 : \tau$ in \onelang then
  $\sem{t_1} = \sem{t_2}$ if, and only if, $\sem{\mathcal{T}(t_1)} =
  \sem{\mathcal{T}(t_2)}$.
\end{Th}
\begin{proof}
  The proof follows from the fact that the translation is a faithful functor.
\end{proof}

It is also possible to translate the multiplicative ($\otimes$, $\lto$) fragment
of \onelang into the \sep-layer of \twolang, by translating the types as
follows:
\begin{mathpar}
  \mathcal{T}'(\bool) \triangleq \M\bool
  \and
  \mathcal{T}'(\tau_1 \otimes \tau_2) \triangleq \mathcal{T}'(\tau_1) \otimes \mathcal{T}'(\tau_2)
  \and
  \mathcal{T}'(\tau_1 \lto \tau_2) \triangleq \mathcal{T}'(\tau_1) \lto \mathcal{T}'(\tau_2)
\end{mathpar}
Once again, the term translation is the identity function.

\begin{Th}
If $\Gamma \vdash t : \tau$ in \onelang then $\mathcal{T}'(\Gamma) \vdashi
\mathcal{T}'(t) : \mathcal{T}'(\tau)$ in \twolang.
\end{Th}
\begin{proof}
  The proof follows by induction on the typing derivation $\Gamma \vdash t : \tau$.
\end{proof}

This translation is functorial and faithful, and therefore is sound and fully
abstract with respect with the denotational semantics of \onelang and \twolang.

\begin{Rem}
It is not possible to translate the whole \onelang into \twolang. Since only one
of the languages of \twolang has arrow types and there is no way of moving from
\sep into \shr, the translation would need to map \onelang programs into \sep
programs, which can only write probabilistically independent programs, making it
impossible to translate the $\times$ type constructor. By adding an additive
function type to the \shr-layer of \twolang, it would be possible to extend the
first translation so that it encompasses the whole language; however, many of
the concrete models that we will consider in the next section do not support an
additive function type in the \shr-layer.
\end{Rem}

\section{Categorical Semantics and Concrete Models}

In this section, we present the general, categorical semantics of \twolang, by
abstracting the probabilistic semantics we saw in the previous section. Then, we
present a variety of concrete models for \twolang, based on existing semantics
for effectful languages. Our soundness theorem ensures natural notions of
separation across these models.

\subsection{Categorical Semantics of \twolang}
\label{sec:catmodel}

Suppose we have two effectful languages, $\mathcal{L}_1$ and $\mathcal{L}_2$.
The first one has a product type $\times$ which allows for the sharing of
resources, while the second one has the disjoint product type $\otimes$.
Furthermore, we assume that $\mathcal{L}_2$ has a unary type constructor $\M$
linking both languages. The intuition behind this decision is that an element of
type $\M \tau$ is a computation which might share resources. From a language
design perspective, the constructor $\M$ serves to encapsulate a possibly
dependent computation in an independent environment.

The first question is to understand is how the connectives $\times$ and
$\otimes$ should be interpreted categorically. For $\times$, we need a
comonoidal structure to duplicate and erase computation. This kind of structure
is captured by \emph{CD categories}, which are monoidal categories where every
object $A$ comes equipped with a commutative comonoid structure $A \to A \otimes
A$ and $A \to I$ making certain diagrams
commute~\citep{DBLP:journals/mscs/ChoJ19}. For $\otimes$, we want to restrict
copying---the separating layer of our language has an affine type system---so
$\otimes$ should be a monoidal product with discard maps.

Finally, to model the type constructor $\M$, the typical categorical idea is
that it should be some kind of functor from $\mathcal{L}_1$ to $\mathcal{L}_2$.
Let us look at some of the intuitions provided by the type system. The type
$\M(\tau_1 \times \tau_2)$ is for computations that may share resources and
output both $\tau_1$ and $\tau_2$.  Meanwhile, the type $\M\tau_1 \otimes
\M\tau_2$ is for computations that output $\tau_1$ and $\tau_2$ while using
separate resources. This reading suggest that there should not be maps from
$\M(\tau_1 \times \tau_2)$ to $\M \tau_1 \otimes \M \tau_2$, since there is no
way of separating resources once they have been shared, but there should be maps
from $\M \tau_1 \otimes \M \tau_2$ to $\M(\tau_1 \times \tau_2)$, since
separation is a specific example of sharing.

Categorically, the existence of these maps is captured by applicative functors,
also known as lax monoidal functors, which are functors $F : (\cat{C},
\otimes_C, I_C) \to (\cat{D}, \otimes_D, I_D)$ between monoidal categories,
equipped with morphisms $\mu_{A, B} : F(A) \otimes_D F(B) \to F(A \otimes_C B)$
and $\epsilon : I_D \to F I_C$ making certain diagrams
commute~\citep{borceux1994}.

Thus, we are led to our categorical model for \twolang.

\begin{Def}
  A \twolang model is a triple $(\cat{C}, \cat{M}, \M)$ where $\cat{C}$ is a
  symmetric monoidal closed category with coproducts and with morphisms $del_A :
  A \to I_C$, natural in $A$; $\cat M$ is a distributive CD category with
  coproducts, i.e., $A \otimes_M (B +_M C) \cong (A \otimes_M B) +_M (A
  \otimes_M C)$; and $\M : \cat{M} \to \cat{C}$ is lax monoidal.
\end{Def}

While we need to assume distributivity in $\cat{M}$, distributivity in $\cat{C}$
holds automatically.

\begin{Lemma}
  In every symmetric monoidal closed category with coproducts, the following
  isomorphism holds: $A \otimes (B + C) \cong (A \otimes B) + (A \otimes C)$. 
\end{Lemma}
\begin{proof}
  By assumption, the functor $A \otimes (-)$ is a left adjoint and, therefore,
  preserves coproducts and we can conclude the isomorphism $A \otimes (B + C)
  \cong (A \otimes B) + (A \otimes C)$.
\end{proof}

The denotational semantics is given in \Cref{fig:twosem} and most of the
equational theory is presented in \Cref{fig:twoeq}. The lax monoidal equations
for $\M$ are long and not very informative, so we leave them to the
Appendix~\ref{app:proof}.





\paragraph*{Soundness.}
In categorical models, the soundness theorem of \twolang can be stated
as follows:

\begin{Th}[Soundness]
\label{th:soundness}
  Let $\cdot \vdash_{I} t : \tau_1 \otimes \tau_2$ then $\sem{t} = f \otimes g$, where $f$ and $g$ are morphisms $I \to \sem{\tau_1}$ and $I \to \sem{\tau_2}$, respectively.
\end{Th}

From a proof-theoretic perspective, the soundness theorem states that for every
proof of type $\cdot \vdash \tau_1 \otimes \tau_2$, we can assume that the last
rule is the introduction rule for $\otimes$. From a semantic perspective, the
soundness theorem means that for every closed term $\cdot \vdash t : \tau_1
\otimes \tau_2$, the semantics $\sem{t}$ can be factored as two morphisms $f_1$
and $f_2$ such that $\sem{t} = f_1 \otimes f_2$.

Establishing soundness requires additional categorical machinery, so we defer
the proof to \Cref{sec:soundness}. Here, we will exhibit a range of concrete
models for \twolang.

\begin{figure}
    \begin{mathpar}
    \inferrule[Var]
    {~}
    {\orange\tau \times \orange\Gamma \xrightarrow[]{id_{\orange\tau} \times del_{\orange\Gamma}}\orange\tau}
\and
    \inferrule[Let]
    {\orange \Gamma \xrightarrow[]{\orange M} \orange \tau_1 \\ \orange \Gamma \times \orange \tau_1 \xrightarrow{\orange N}  \orange \tau_2}
    {\orange \Gamma \xrightarrow[]{copy; (id \times \orange M); \orange N} \orange \tau_2}
\and
    \inferrule[$\times$ Intro]
    {\orange\Gamma \xrightarrow[]{\orange M} \orange{\tau}_1 \\ \orange\Gamma \xrightarrow[]{\orange N} \orange{\tau}_2}
    {\orange\Gamma \xrightarrow[]{copy; \orange M \times \orange N} \orange{\tau}_1 \times \orange{\tau}_2}
\and
    \inferrule[$\times$ Elim$_i$]
    {\orange\Gamma \xrightarrow[]{\orange M} \orange \tau_1 \times \orange \tau_2}
    {\orange \Gamma \xrightarrow{\orange M; (id_{\orange \tau_i} \times del)} \orange \tau_i}
\\
    \inferrule[$+$ Intro$_i$]
    {\orange \Gamma \xrightarrow[]{\orange M} \orange \tau_1}
    {\orange \Gamma \xrightarrow[]{\orange M; in_i} \orange \tau_1 + \orange \tau_2}
\and
    \inferrule[$+$ Elim]
    {\orange \Gamma_1 \xrightarrow[]{\orange N}  \orange \tau_1 + \orange \tau_2 \\ \orange \Gamma_2 \times \orange \tau_1 \xrightarrow[]{\orange M_1} \orange \tau \\ \orange \Gamma_2 \times \orange \tau_2 \xrightarrow[]{\orange M_2} \orange \tau}
    {\orange \Gamma_1,\orange \Gamma_2 \xrightarrow[]{\orange N \times id_{\orange \Gamma_2}} (\orange \tau_1 + \orange \tau_2) \times\orange \Gamma_2 \cong (\orange \tau_1 \times \orange \Gamma_2) + (\orange \tau_2 \times \orange \Gamma_2) \xrightarrow{[\orange M_1, \orange M_2]} \orange \tau}
\\
\\
    \inferrule[Var]
    {~}
    {\purple \Gamma, \purple {\underline{\tau}} \xrightarrow[]{del_{\purple \Gamma} \otimes id_{\purple \tau}} \purple {\underline{\tau}}}
\and 
    \inferrule[Abstraction]
    {\purple \Gamma \otimes  \purple{\underline {\tau_1}} \xto{\purple t} \purple{\underline {\tau_2}}}
    {\purple \Gamma \xto{\mathsf{cur}(\purple t)} \purple{\underline {\tau_1}} \lto \purple{\underline {\tau_2}}}
\and
    \inferrule[Application]
    {\purple \Gamma_1 \xto{\purple t} \purple{\underline {\tau_1}} \lto \purple{\underline {\tau_2}} \\ \Gamma_2 \xto{u} \purple{\underline {\tau_1}}}
    {\purple \Gamma_1\otimes \purple \Gamma_2 \xto{(\purple t \otimes \purple u); \mathsf{ev}} \purple {\underline{\tau_2}}}
\\
    \inferrule[$\otimes$ Intro]
    {\purple \Gamma_1 \xrightarrow[]{\purple t} \purple {\underline{\tau_1}} \\ \purple \Gamma_2 \xrightarrow[]{\purple u} \purple {\underline{\tau_2}}}
    {\purple \Gamma_1 \otimes  \purple \Gamma_2 \xrightarrow[]{\purple t \otimes \purple u} \purple {\underline{\tau_1}} \otimes \purple {\underline{\tau_2}}}
\and
    \inferrule[$\otimes$ Elim]
    {\purple \Gamma_1 \xto{\purple t} \purple{\underline {\tau_1}} \otimes \purple{\underline {\tau_2}} \\ \purple \Gamma_2 \otimes \purple{\underline {\tau_1}} \otimes \purple{\underline {\tau_2}} \xto{\purple u} \purple{\underline \tau}}
    {\purple \Gamma_1 \otimes \purple \Gamma_2 \xto{(id \otimes \purple t); \purple u} \purple{\underline \tau}}
\\
    \inferrule[$\oplus$ Intro$_i$]
    {\purple \Gamma \xrightarrow[]{\purple t} \purple {\underline{\tau_i}}}
    {\purple \Gamma \xrightarrow[]{\purple t; in_i} \purple {\underline{\tau_1}} + \purple {\underline{\tau_2}}}
\and
    \inferrule[$\oplus$ Elim]
    {\purple \Gamma_1 \xrightarrow[]{\purple u}  \purple {\underline{\tau_1}} + \purple {\underline{\tau_2}} \\ \purple {\underline{\tau_1}} \otimes \purple \Gamma_2 \xrightarrow[]{\purple t_1} \purple {\underline{\tau}} \\ \purple {\underline{\tau_2}} \otimes \purple \Gamma_2 \xrightarrow[]{\purple t_2} \purple {\underline{\tau}}}
    {\purple \Gamma_1, \purple \Gamma_2 \xrightarrow[]{\purple u \otimes id_{\purple \Gamma_2}} (\purple {\underline{\tau_1}} + \purple {\underline{\tau_2}}) \otimes\purple \Gamma_2 \cong (\purple {\underline{\tau_1}} \otimes \purple \Gamma_2) + (\purple {\underline{\tau_2}} \otimes \purple \Gamma_2) \xrightarrow{[\purple t_1, \purple t_2]} \purple {\underline{\tau}}}
\\\\
    \inferrule[Sample]
    {\orange \tau_1 \times \cdots \times \orange \tau_n \xto{\orange M} \orange \tau \\ \purple \Gamma_i \xto{\purple t_i} \M \orange \tau_i}
    {\purple \Gamma_1 \otimes \cdots \otimes \purple \Gamma_n \xto{\purple t_1
    \otimes \cdots \otimes \purple t_n} \M \orange \tau_1 \otimes \cdots \otimes
  \M \orange \tau_n \xto{\mu} \M (\orange \tau_1 \times \cdots \times \orange \tau_n)\xto{\M \orange M} \M \orange \tau}
  \end{mathpar}
\hrule
  \caption{Categorical Semantics: \twolang}
  \label{fig:twosem}
\end{figure}


\subsection{Concrete models}
\label{sec:models}

To warm up, we present some basic probabilistic models \twolang. While prior
work has also investigated similar models~\citep{azevedodeamorim2022sampling},
we adapt these models to \twolang and explain how our soundness theorem ensures
independence.

\begin{figure}
    \begin{align*}
    \caseof {(\mathsf{in}_1 M)} {N_1} {N_2}  &\;\;\equiv\;\; \subst{N_1}{x}{M}
    \\
    \caseof {(\mathsf{in}_2 M)} {N_1} {N_2}  &\;\;\equiv\;\; \subst{N_2}{x}{M}
    \\[1.5ex]
    \letin{x}{t}{x} &\;\;\equiv\;\; t
    \\
    \letin{x}{x}{t} &\;\;\equiv\;\; t
    \\
    \letin{y}{(\letin{x}{M_1}{M_2})}{{M_3}} &\;\;\equiv\;\; \letin{x}{M_1}{(\letin{y}{M_2}{M_3}})
    \\
    \\
    \app {(\lamb x t)} u &\;\;\equiv\;\; \subst t x u
    \\
    \letin{x_1 \otimes x_2}{t_1 \otimes t_2}{u} &\;\;\equiv\;\; \subst {\subst u {x_1}{t_1}}{x_2}{t_2}
    \\[1.5ex]
    \caseof {(\mathsf{in}_1 t)} {u_1} {u_2}  &\;\;\equiv\;\; \subst{u_1}{x}{t}
    \\
    \caseof {(\mathsf{in}_2 t)} {u_1} {u_2}  &\;\;\equiv\;\; \subst{u_2}{x}{t}
    \\[1.5ex]
    \sample t x x &\;\;\equiv\;\; t
    \\
    \sample {(\sample t x M)} y {N} &\;\;\equiv\;\; \sample t x {(\letin y M N)}
  \end{align*}
\hrule
  \caption{Equational Theory: \twolang}
  \label{fig:twoeq}
\end{figure}

\subsubsection{Discrete Probability}

Our first concrete model is a different semantics for discrete probability. For
the sharing category, we take the category $\cat{CountStoch}$ with countable
sets as objects, and transition matrices as morphisms, i.e. functions $f : A
\times B \to [0,1]$ such that for every $a \in A$, $f(a, -)$ is a (discrete)
probability distribution~\citep{fritz2020}.

For the independent category, we take the probabilistic coherence space model of
linear logic, a well-studied semantics for discrete probabilistic
languages~\citep{pcoh}. This model was originally used to explore the connections
between probability theory and linear logic, and has recently been used to
interpret recursive probabilistic programs and recursive
types~\citep{tasson2019}; it is also fully-abstract for probabilistic PCF
\citep{DBLP:journals/jacm/EhrhardPT18}.  

\begin{Def}[\citet{pcoh}]
  A \emph{probabilistic coherence space (PCS)} is a pair $(|X|, \Pcal (X))$ where $|X|$
  is a countable set and $\Pcal (X) \subseteq |X| \to \R^+$ satisfies:

  \begin{itemize}
  \item $\forall a \in |X|\ \exists \varepsilon_a > 0\ \varepsilon_a \cdot \delta_a \in \mathcal P (X)$, where $\delta_a(a') = 1$ iff $a = a'$ and $0$ otherwise;
  \item $\forall a \in |X|\ \exists \lambda_a\ \forall x \in \mathcal P (X)\ x_a \leq \lambda_a$;
  \item $\mathcal P (X)^{\perp\perp} = \mathcal P (X)$, where $\mathcal P (X)^\perp = \set {x \in |X| \rightarrow \R^+} {\forall v \in \mathcal P(X)\ \sum_{a \in |X|}x_av_a \leq 1}$.
  \end{itemize}

\end{Def}

We can define a category $\cat{PCoh}$ where objects are probabilistic coherence
spaces and morphisms $X \lto Y$ are matrices $f : |X| \times |Y| \to \R^+$ such
that for every $v \in \Pcal{(X)}$, $f\, v \in \Pcal{(Y)}$, where $(f \, v)_b =
\sum_{a \in |X|}f_{(a,b)}v_a$.  It is well-known that this category is a SMCC; we
will use the explicit definition of the monoidal product.

\begin{Def}
\label{def:monpcoh}
  Let $(|X|, \Pcal{(X)})$ and $(|Y|, \Pcal{(Y)})$ be PCS, we define $X \otimes Y
  = (|X| \times |Y|, \set{x \otimes y}{ x \in \Pcal{(X)}, y \in
  \Pcal{(Y)}}^{\perp\perp})$, where $(x \otimes y)(a, b) = x(a)y(b)$.
\end{Def}

We can now define a functor $\M : \cat{CountStoch} \to \cat{PCoh}$.

\begin{Lemma}[see, e.g., \citet{azevedodeamorim2022sampling}]
  \label{lem:natpcs}
  Let $X$ be a countable set, the pair $(X, \set{\mu : X \to \R^+}{\sum_{x \in
  X} \mu(x) \leq 1})$ is a PCS. Any $\cat{CountStoch}$ morphism $X \to Y$ is
  also a $\cat{PCoh}$ morphism.
\end{Lemma}

\begin{Th}
  There is a lax monoidal functor $\M : \cat{CountStoch} \to \cat{PCoh}$.
\end{Th}
\begin{proof}
  The functor is defined using the previous above. Functoriality holds because
  the functor is the identity on arrows. The lax monoidal structure is given by
  $\epsilon = id_{1}$ and $\mu_{X, Y} = id_{X \times Y}$.
\end{proof}

Summing up, we have a model of \twolang based on probabilistic coherence spaces.

\begin{Th}
  The triple $(\cat{ \cat{PCoh}, CountStoch}, \M)$ is a \twolang model.
\end{Th}
\begin{proof}
  $\cat{CountStoch}$ is well-known to be a CD category with coproducts
  \citep{fritz2020}, and $\cat{PCoh}$ is a symmetric monoidal closed category
  with coproducts because it is a model of linear logic~\citep{pcoh}. The
  morphism $del_X$ is given by the constant $0$ function, where the monoidal
  unit is the interval $[0,1]$.  Finally, lax monoidality of $\M$ is given by
  the previous theorem. 
\end{proof}

In $\cat{PCoh}$ it is possible to show that $\M \tau_1 \otimes \M \tau_2
\subseteq \M (\tau_1 \times \tau_2)$ meaning that well typed programs of type
$\M \tau_1 \otimes \M \tau_2$ are denoted by joint distributions over $\tau_1
\times \tau_2$. Furthermore, by taking a closer look at Definition~\ref{def:monpcoh} 
we see that $\mu_A \otimes \mu_B$ corresponds exactly to the product 
distribution of $\mu_A$ and $\mu_B$, so our soundness theorem
implies that closed programs of type $\M \tau_1 \otimes \M \tau_2$ are
denoted by independent probability distributions.

\subsubsection{Continuous Probability}

Next, we consider models for continuous probability. For the sharing layer, the
generalization of $\cat{CountStoch}$ to continuous probabilities is
$\cat{BorelStoch}$, which has standard Borel spaces as objects and Markov
kernels as morphisms~\citep{fritz2020}; see Appendix~\ref{app:meas} for details.
For the separating layer, we want a model of linear logic that can interpret
continuous randomness. We use a model based on perfect Banach lattices.



\begin{Def}[\citet{azevedodeamorim2022riesz}]
  The category $\cat{PBanLat_1}$ has perfect Banach lattices as objects and
  order-continuous linear functions with norm at most one as morphisms.
\end{Def}

Intuitively, a perfect Banach lattice is a Banach space equipped with a lattice
structure and an involutive linear negation. For every measurable space $(X,
\Sigma_X)$ the space of signed measures over it is a perfect Banach space,
meaning that it can, for instance, interpret continuous probability
distributions over the real line. Furthermore, the map assigning $(X, \Sigma_X)$
to its space of signed measures is functorial and lax monoidal.

\begin{Th}[\citet{azevedodeamorim2022riesz}]
  There is a lax monoidal functor $\M : \cat{BorelStoch} \to \cat{PBanLat_1}$.
\end{Th}

\begin{Th}
  The triple $(\cat{PBanLat_1}, \cat{BorelStoch}, \M)$ is a \twolang model.
\end{Th}
\begin{proof}
  The category $\cat{BorelStoch}$ has a CD structure and has coproducts because
  it is isomorphic to the Kleisli category of a commutative monad over the
  category $\cat{Meas}$~\citep{fritz2020}. The category $\cat{PBanLat_1}$ is a
  model of classical linear logic, making it a SMCC with
  coproducts~\citep{azevedodeamorim2022riesz}. The morphism $del_V$ is the
  constant $0$ function, where the monoidal unit is $\R$. The lax monoidality of
  $\M$ follows from the previous theorem. 
\end{proof}

This model can be seen as the continuous generalization of the previous model,
since there are full and faithful embeddings $\cat{CountStoch}
\hookrightarrow \cat{BorelStoch}$ and $\cat{PCoh} \hookrightarrow
\cat{PBanLat}_1$ \citep{azevedodeamorim2022riesz}.  In this model, our soundness
theorem once again ensures probabilistic independence, i.e. programs of type $\M
\tau_1 \otimes \M \tau_2$ are denoted by independent distributions.

\subsubsection{Non-Determinism and Communication}

Next, we show that the relational model of linear logic gives rise to a \twolang
model, with applications with distributed programming.

\paragraph*{Semantics}
Our starting point is the category $\cat{Rel}$ of sets and binary relations, one
of the most well-known models for linear logic. By pairing this category with
the Kleisli category $\cat{Set}_{\Pcal}$, for the powerset monad $\Pcal$ we 
immediately obtain a model for \twolang.

\begin{Th}
    The triple $(\cat{Rel}, \cat{Set}_{\Pcal}, id)$ is a \twolang model.
\end{Th}
\begin{proof}
    Binary relations over sets $A$ and $B$ are represented either as subsets $R \subseteq A \times B$ or,
    equivalently, as functions $A \to \Pcal(B)$. From this observation it is possible to show that the
    identity functor is an isomorphism and it easily follows from this that $id$ is lax monoidal. Since
    $\cat{Rel}$ is a model of linear logic, it has coproducts and, by isomorphism, so does $\cat{Set}_{\Pcal}$.
    The natural transformation $del$ in $\cat{Rel}$ is the delete operation from $\cat{Set}_{\Pcal}$.
\end{proof}

\paragraph*{Application to Distributed Programming}
While this model arises from linear logic, we show that it leads to a suitable
language for distributed programming.  We assume a two-tier approach to
programming with communication: the \shr language is used for writing local
programs, while the \sep language is used to orchestrate the communication
between local code. Programs of type $\M \textcolor{purple}{\underline{\tau}}$
correspond to local computations that can be manipulated by the communication
language. Programs in the $\sep$ language are interpreted as maps of the form $A
\to \Pcal(B)$; we view these maps as allowing \emph{non-deterministic} or
\emph{lossy} communication.

To align the syntax with this interpretation, we tweak the syntax $\sample {t_i}
{x_i} {M}$ to $\mathsf{send}\, t_i \, \mathsf{as}\,  x_i\, \mathsf{in}\, M$
which sends the values computed by the local programs $t_i$, binds them to $x_i$
and continues as the local program $M$. To see how how distributed programs can
be written in this language, we consider a simple distributed voting protocol
between two parties. We suppose that there is a leader that receives two
messages containing the votes and if they are the same, the election is decided
and the leader announces the winner. If the votes disagree, the leader outputs a
tagged unit value saying that there has been a draw. In \twolang, the leader can
be implemented as:
\begin{align*}
\mathsf{\purple {leader}} &: \M \orange \nat \otimes \M \orange \nat \lto \M(\orange{\nat \oplus 1})\\
\mathsf{\purple {leader}} &= \purple \lambda \, \purple{x_1\, x_2}.\, \mathsf{send}\, \purple {x_1, x_2}\, \mathsf{as} \, \orange{n_1, n_2} \, \mathsf{in} \, \orange{\ifthenelse{n_1 = n_2}{(\mathsf{in_1}\, n_1)}{(\mathsf{in_2}\, ())}}
\end{align*}
Given a program $\purple {\mathsf{votes}} : \M \orange \nat \otimes \M \orange
\nat$ that computes what each agent will vote, the full distributed program can
be represented as the application $\purple{\app
{\mathsf{leader}}{\mathsf{votes}}}$.

\paragraph*{Soundness theorem}
In this model, our soundness result ensures that if we have a closed program of
type $\M \tau_1 \otimes \M \tau_2$, then it can be factored as two local
programs that can be run locally, and do not require any extra communication
other than the $\mathsf{send}$ instructions.  To understand why this guarantee
is non-trivial, consider the problematic program from
Section~\ref{sec:twolevel}:
\[
\purple{\mathsf{message}} : \M(\purple{1 + 1}) \nvdash_{\sep}\, \purple{\ifthenelse{\mathsf{message}}{(\kwtt \otimes \kwtt)}{(\kwff \otimes \kwff)}} : \M \purple\bool \,\purple\otimes\, \M \purple\bool
\]
Under our interpretation, the if-statement is conditioning on the contents of
the program variable $\purple{\mathsf{message}}$ and producing two local computations
that have the same outputs. There are two potential sources of implicit
communication in this program. First, the contents of $\purple{\mathsf{message}}$ are
non-deterministic, so the local computations must communicate in order to agree
on what value to return. Second, by conditioning on the same value, the message
must be sent to both local computations.  These indirect communications have
already been addressed in the choreography literature, as illustrated by
\citet{hirsch2022}, where their language allows pattern matching on local
computation but the chosen branch must be broadcast to programs that depend on
it, which is not problematic in a setting where communication is reliable.

To illustrate the soundness guarantee, we can revisit the distributed voting
example. By the soundness theorem, the program $\purple{\mathsf{votes}}$ is
equal to $\purple {t_1 \otimes t_2}$ for programs $\purple {t_1}, \purple{t_2}
: \M \orange \nat$. Thus, the only communication required are explicit sends.

\paragraph{Expressivity and Limitations}
Intuitively, closed programs in \twolang of type $\M \tau$ are equivalent to
$\send{t_i}{x_i}{M}$, which we view as a local program $M$ that starts by
receiving $n$ different messages, runs its body $M$ with the received messages
as bound variables, and makes its output available to be sent to a different
local computation.  Therefore, each local program may only have one block of
receives at the beginning and one send at the end, limiting the allowed
communication patterns.

These limitations have been addressed in other modal logic approaches to
distributed programming by having a static set of agents and a modality
annotated by elements of this set representing computations that are executed by
a particular agent of the distributed system~\citep{hirsch2022}. We conjecture that by
extending \twolang with type constructors $\M_\ell \tau$, where $\ell$ is the
name of an agent, it might be possible to represent more intricate communication
patterns, but we leave this for future work. 

\paragraph{Related Work}

Distributed programming is challenging and error-prone, and there is a long
history of language design in this setting. Two notable examples are session
types~\citep{session-types} and choreographic
programming~\citep{choreographies}. Session types adopts a linear typing
discipline where type constructors model the desired protocol.  On the other
hand, choreographic programming adopts a monolithic approach: The entire system
is written as a single program that can be compiled to ``local computations'',
with the compiler adding the appropriate communication instructions.

Our model of \twolang blends aspects of both approaches. It still has
a substructural communication type system, but it also represents protocols using a
single global program with a two-tier language that distinguishes between local
and global computation. We leave a more thorough comparison between these
languages for future work.

\subsubsection{Commutative Effects}

In this section we will present a large class of models based on commutative monads which are monads where, in a Kleisli semantics of effects, the program equation $(\letin{x}{t}{\letin{y}{u}{w}}) \equiv
 (\letin{y}{u}{\letin{x}{t}{w}})$ holds. 


The Kleisli category of commutative monads has many useful properties.

\begin{Th}[\citet{fritz2020}]
  Let $\cat{C}$ be a Cartesian category and $T$ a commutative monad over it. The category $\cat{C}_T$ is a CD category.
\end{Th}

\begin{Lemma}
  Let $\cat{C}$ be a distributive category and $T$ a monad over it. Its Kleisli category $\cat{C}_T$ has coproducts and is also distributive. 
\end{Lemma}
\begin{proof}
  It is straightforward to show that Kleisli categories inherit coproducts from the base category. Furthermore, by using the distributive structure of $\cat{C}$, applying $T$ to it and using the functor laws, it follows that $\cat{C}_T$ is distributive.
\end{proof}

Another useful category of algebras is the category of algebras and plain maps $\widetilde{\cat{C}^T}$ which has $T$ algebras as objects and $\widetilde{\cat{C}^T}((A, f), (B, g)) = \cat{C}(A, B)$.

\begin{Th}[\citet{simpson1992}]
  Let $\cat{C}$ be a Cartesian closed category and $T$ a commutative monad over it. The
  category of $T$-algebras and plain maps is Cartesian closed, and $1$ is a
  terminal object.
\end{Th}

Therefore, we choose the Kleisli category to interpret \shr and the category of $T$-algebras and plain maps to interpret \sep. We only have to show that there is an applicative functor between them.

\begin{Th}
  There exists an applicative functor $\iota : \cat{C}_T \to \widetilde{\cat{C}^T}$.
\end{Th}
\begin{proof}
  The functor acts by sending objects $A$ to the free algebra $(TA, \mu_A)$ and morphisms $f : A \to TB$ to $f^*$. Now, for the lax monoidal structure, consider the natural transformation $\mu \circ T \tau \circ \sigma : T A \times T B \to T(A \times B)$ and $\eta_1 : 1 \to T 1$, where $\tau$ and $\sigma$ are the strengths of $T$. It is possible to show that this corresponds to an applicative functor by using the fact that $T$ is commutative and that the comonoid structure $A \to 1$ is natural.
\end{proof}



\begin{Th}
\label{th:kleisli}
  The triple $(\widetilde{\cat{C}^T}, \cat{C}_T, \iota)$ is a \twolang model.
\end{Th}

\paragraph{Name generation}
Simple concrete examples of commutative effects are probability and
non-determinism, which we saw before.  A more interesting example is the name
generation monad used to give semantics to the $\nu$-calculus, a language that
has a primitive for generating ``fresh'' symbols~\citep{stark1996}. This is a
useful abstraction, for instance, in cryptography, where a new symbol might be a
secret that you might not want to share with adversaries.

A concrete semantics to the $\nu$-calculus was presented by \citet{stark1996}
where the base category is the functor category $[\cat{Inj}, \cat{Set}]$, with
$\cat{Inj}$ being the category of finite sets and injective functions. In this case
the (commutative) name generation monad acts on functors as 
\[
T(A)(s) = \set{(s', a')}{s' \in \cat{Inj}, a' \in A(s + s')}/\sim
\]
where $(s_1, a_1) \sim (s_2, a_2)$ if, and only if, for some $s_0$ there are
injective functions $f_1 : s_1 \to s_0$ and $f_2 : s_2 \to s_0$ such that
$A(id_s + f_1)a_1 = A(id_s + f_2)a_2$. The intuition is that $T(A)$ is a
computation that, given a finite set $s$ of names used, produces the newly
generated names $s'$, and a value $a'$. By Theorem~\ref{th:kleisli} the triple
$(\widetilde{[\cat{Inj}, \cat{Set}]^T}, [\cat{Inj}, \cat{Set}]_T, \iota)$ is a
\twolang model.

Syntactically, we can extend the type grammar of the $\shr$ language with a type
$\mathsf{Name}$ for names, and the $\shr$ language with an operation $\cdot
\vdash \mathsf{fresh} : \mathsf{Name}$ for name generation.  Our soundness
theorem says that for a program of type $\M \tau \otimes \M \tau$, the names
used to compute the first component are \emph{disjoint} from the ones used to
compute the second component.



It is also possible to define a variant to this algebra model using the Eilenberg-Moore
category since this category is known to be symmetric monoidal closed under a
few minor hypothesis \citep{azevedodeamorim2022sampling}.





\begin{Rem}[Call-by-Value and Call-by-Name Semantics of Effects]

Categories of algebras and plain maps were used as a denotational foundation for
call-by-name programming languages while Kleisli categories can be used to
interpret call-by-value languages~\citep{simpson1992}. Thus, the \sep language can be
seen as a CBN interpretation of effects, while \shr can be seen as a CBV
interpretation of effects. The operational interpretation of
$\sample{\textcolor{purple}{\overline{t}}}{\textcolor{orange}{\overline{x}}}{\textcolor{orange}{M}}$
is to force the execution of CBN computations $\overline{t}$, bind the results
to $\overline{x}$, and run them eagerly in the program $M$.
\end{Rem}

\subsubsection{Affine Bunched Typing}

The logic of bunched implications (BI)~\citep{DBLP:journals/bsl/OHearnP99} is a
substructural logic, developed for reasoning about sharing and separation of
resources like pointers to a heap memory~\citep{DBLP:conf/csl/OHearnRY01}, or
permissions to enter some critical section in concurrent
code~\citep{DBLP:conf/iwmm/OHearn07}.  The proof theory of BI gives rise to
functional languages with bunched type systems, where contexts are trees
(so-called \emph{bunches}) rather than lists~\citep{DBLP:journals/jfp/OHearn03}.


It is natural to wonder how BI is related to \twolang. Semantically, bunched
calculi are interpreted using a \emph{doubly closed category} (DCC), a single category that has both a Cartesian closed
and a (usually distinct) monoidal closed structure.  In order to understand how
these systems are related, let us consider the affine variant of the bunched
calculus, i.e., when the monoidal unit is a terminal object in the semantic category, meaning that there is a
discard operation $A \otimes B \to A$. Given an affine BI model $\cat{C}$,
there is a morphism $A \otimes B \to A \times B$ given by the universal property of products applied to the discard morphisms $A\otimes B \to A$ and $A \otimes B \to B$. Furthermore,
by assumption $I \cong 1$, where $1$ is the unit for the Cartesian product and $I$ is the unit for the monoidal product.
Finally, such a structure makes the lax monoidality diagrams commute,
making the identity functor $id : (\cat{C}, \times, 1) \to (\cat{C}, \otimes, I)$ a lax monoidal functor between
the two monoidal structures over $\cat{C}$. Thus:

\begin{Th}
\label{th:ABItwolang}
  For every model of affine BI $\cat{C}$ the triple $(\cat C, \cat C, id)$ is a model of \twolang.
\end{Th}

\begin{Rem}
From a more abstract point of view, by initiality of the syntactic model of
\twolang (Theorem~\ref{th:initial}) and the theorem above,
there is a translation from \twolang to the
bunched calculus. Thus, affine bunched calculi can be seen as a degenerate
version of our language, where the two layers are collapsed into one.
\end{Rem}

\paragraph*{Syntactic Control of Interference}
To illustrate a useful model of the affine bunched calculus, let us consider
O'Hearn's bunched language SCI+~\citep{DBLP:journals/jfp/OHearn03}. This
language allows allocating memory and reasoning about aliasing, building on
Reynolds' Syntactic Control of Interference (SCI), a linear type system.
In the denotational semantics of SCI+, types are objects in the functor category
$\cat{Set}^{\Pcal(Loc)}$, where $\Pcal(Loc)$ is the poset category of subsets of
$Loc$, an infinite set of names (i.e., memory addresses).  Intuitively, a
presheaf maps a subset of locations to the set of computations that use those
locations.  It is well-known that this category is a model of affine BI: The
Cartesian closed structure is given by the usual construction on presheaves,
while the monoidal closed structure is given by a different product on
presheaves, called the Day convolution~\citep{borceux1994}. 

By Theorem~\ref{th:ABItwolang} the triple $(\cat{Set}^{\Pcal(Loc)},
\cat{Set}^{\Pcal(Loc)}, id)$ is a \twolang model and, therefore, satisfies its
soundness property. To understand what it means in this context, we look at how
the model is defined. Given presheaves $A$ and $B$ over $\Pcal(Loc)$, the
monoidal product $A \otimes B$ is defined as
\begin{align*}
    (A \otimes B)(X) &\triangleq \set{(a, b) \in A(X) \times B(X)}{support(a) \cap support(b) = \emptyset}\\
    (A \otimes B)(f) &\triangleq (Afa,Bfb)
\end{align*}
The $support$ function acts on sets and has a slightly technical definition that
models which resources in $Loc$ were used to produce the set---the interested
reader should consult the original paper~\citep{DBLP:journals/jfp/OHearn03}. At
a high level, the disjointness of the support captures the fact that the memory
locations used to produce $a$ are disjoint from the memory locations used to
produce $b$. Therefore, our soundness theorem guarantees that the components of
closed programs of type $\M \tau_1 \otimes \M \tau_2$ do not share any memory
locations. 

At the syntactic level, the SCI+ calculus shares some similarities with \onelang,
such as having two distinct product types, but it also has many differences. For
instance it has two context concatenation operations, making it possible to
accommodate two different kinds of arrow types, shown in
Figure~\ref{fig:scilang}. Additionally, it features ground types $\mathsf{exp}$,
$\mathsf{cell}$ and $\mathsf{comm}$ for expressions, memory cells and commands,
respectively, and primitive operations to manipulate them.
\begin{figure}[]
  \begin{syntax}
    \category[types]{\tau}
    \alternative{\mathsf{cell}}
    \alternative{\mathsf{exp}}
    \alternative{\mathsf{comm}}
    \alternative{\tau \to \tau}
    \alternative{\tau \lto \tau}
    \alternative{\tau \times \tau}

    \separate

    \category[contexts]{\Gamma}
    \alternative{\cdot}
    \alternative{x : \tau}
    \alternative{\Gamma;\Gamma}
    \alternative{\Gamma,\Gamma}
  \end{syntax}
\hrule
  \caption{Types and Terms: SCI+}
  \label{fig:scilang}
\end{figure}

For our purposes, we are mainly interested in the SCI+ operations presented in
Figure~\ref{fig:scirules}.  The first two rules are for composing commands
either sequentially or in parallel, respectively.  The following two rules are
the ones related to memory manipulation, where the first one allocates a new
memory location and the second one assigns a value to a location. The final two
are the two applications: the first allows the context to be shared, while the
second does not.

A notorious difficulty of running stateful programs in parallel is that there
might be concurrent writes to the same memory location. This is avoided in SCI+
by using the separating concatenation of contexts, guaranteeing that no such
conflict of writes can occur.  When programs are sequentially composed, no such
issues come up and the context may be shared.  When a new memory cell is
allocated using the $\mathsf{new}\, x.M$ syntax, a new variable is bound to the
context representing the new location which is disjoint from the existing ones,
hence the separating context extension.

\begin{figure}
\begin{mathpar}
  \inferrule[]{\Gamma \vdash M : \mathsf{comm} \\ \Gamma \vdash N : \mathsf{comm}}{\Gamma \vdash M; N : \mathsf{comm}}
  \and
  \inferrule[]{\Gamma_1 \vdash M : \mathsf{comm} \\ \Gamma_2 \vdash N : \mathsf{comm}}{\Gamma_1, \Gamma_2 \vdash M || N : \mathsf{comm}}
  \\
  \inferrule[]{\Gamma, x : \mathsf{cell} \vdash M : \mathsf{comm}}{\Gamma \vdash \mathsf{new}\, x. M : \mathsf{comm}}
  \and
  \inferrule[]{\Gamma \vdash M : \mathsf{cell} \\ \Gamma \vdash N : \mathsf{exp}}{\Gamma \vdash M := N : \mathsf{comm}}
  \\
  \inferrule[]{\Gamma \vdash M : \tau_1 \to \tau_2 \\ \Gamma \vdash N : \tau_1}{\Gamma \vdash \app{M}{N} : \tau_2}
  \and
  \inferrule[]{\Gamma_1 \vdash M : \tau_1 \lto \tau_2 \\ \Gamma_2 \vdash N : \tau_1}{\Gamma_1,\Gamma_2 \vdash \app{M}{N} : \tau_2}
\end{mathpar}
\hrule
\caption{Typing Rules: SCI+ (selected)}
\label{fig:scirules}
\end{figure}

\paragraph*{SCI+ in \twolang}
As we have explained, a direct consequence of Theorem~\ref{th:ABItwolang} is
that there is a translation of \twolang into the BI calculus.  However, it is
not a direct consequence that the cell and command operations can be given
similar typing rules and semantics to their original formulation. By slightly
modifying \twolang we can accommodate them as we show in
Figure~\ref{fig:twolangsci}.  Sequential composition is done in the \shr
language while parallel composition is done at the \sep language. The cell
assignment rule is added to the \shr language, since there is no reason to
require that a cell's address and its value are computed using separate
locations. For cell allocation, the original rule requires the new cell to be
disjoint from the existing ones, making it natural to use the \sep language.


\begin{figure}
\begin{mathpar}
  \inferrule[Sequential]{\orange\Gamma \vdashni \orange M : \orange {\mathsf{comm}} \\ \orange\Gamma \vdashni \orange N : \orange{\mathsf{comm}}}{\orange\Gamma \vdashni \orange {M; N} : \orange{\mathsf{comm}}}
  \and
  \inferrule[Parallel]{\purple{\Gamma_1} \vdashi \purple t : \M \orange {\mathsf{comm}} \\ \purple\Gamma \vdashi \purple u : \M \orange{\mathsf{comm}}}{\purple{\Gamma_1}, \purple{\Gamma_2} \vdashi \purple {t \, ||\, u} : \M \orange{\mathsf{comm}}}
  \\
  \inferrule[New]{\purple \Gamma, \purple x : \M \orange{\mathsf{cell}} \vdashi \purple t : \M \orange{\mathsf{comm}}}{\purple \Gamma \vdashi \purple{\mathsf{new}\, x. t} : \M \orange{\mathsf{comm}}}
  \and
  \inferrule[Assign]{\orange \Gamma \vdashni \orange M : \orange{\mathsf{cell}} \\ \orange\Gamma \vdashni \orange N : \orange{\mathsf{exp}}}{\orange\Gamma \vdashni \orange{M := N} : \orange{\mathsf{comm}}}
\end{mathpar}
\hrule
\caption{Typing Rules: \twolang extended with SCI primitives}
\label{fig:twolangsci}
\end{figure}

\begin{Ex}[\citet{DBLP:journals/jfp/OHearn03}]
Consider the \twolang program $\app {\app {(\lamb {x\, y}{x := 1; y:=2})} {z}}
{z}$. There are two possible types for the $\lambda$-abstraction.  The
type $\M\orange{\mathsf{cell}} \purple \lto \M\orange{\mathsf{cell}} \purple
\lto \M\orange{\mathsf{comm}}$ requires that the input locations $x$ and $y$
must be disjoint, while the type $\M\orange{(\mathsf{cell} \times
\mathsf{cell})} \purple \lto \M\orange{\mathsf{comm}}$ allows $x$ and $y$ to be
shared.
The former makes the application ill-typed, since the arguments to the
abstraction are the same, while the latter is well-typed. Note, however, that it
is only well-typed because the assignments are sequentially composed. If they
were composed in parallel the program would be ill-typed, just like in
SCI+, since parallel composition requires disjoint memory locations.
\end{Ex}

\paragraph*{A more expressive \twolang}
SCI+ supports more fine-grained sharing/disjointness policies that interleave
the $\times$ and $\otimes$ type constructors---these programs are difficult to
express in \twolang.  For instance, it is not possible to represent the type
$\M(A \otimes B) \times \M (C \otimes D)$ in our language. This limitation is
because there is only one modality mapping the \shr language into the \sep
language, and no modality going the other way. This limitation can also be seen
in the following simple program, which cannot be expressed in \twolang: $x := 1;
(y := 2) \, || \, (z := 3)$. The program is ill-typed because only \shr programs
can be sequentially composed and only \sep programs can be composed in parallel.
In the concrete model, however, the lax monoidal functor is the identity
functor, allowing us to add the clause $\orange{\tau} \coloneqq \purple{\tau}
\smid \cdots$ to the \shr type grammar and making the following typing rule
sound:
\[
\inferrule[]{\purple \Gamma \vdashi \purple t : \purple \tau}{\orange \Gamma \vdashni \purple t : \purple \tau }
\]
which makes it possible to type check the troublesome program above.

\section{Soundness Theorem}
\label{sec:soundness}

So far we have seen two proofs of soundness. For \onelang, we proved soundness
using logical relations~(Theorem~\ref{th:soundkl}). For \twolang with a
probabilistic semantics, we used an observation about algebras for the
distribution monad~(Theorem~\ref{th:soundem}). This proof is slick, but the
strategy does not generalize to other models of \twolang.

Thus, to prove our general soundness theorem for \twolang, we will return to
logical relations. The statement of our soundness theorem is as follows.

\begin{Th}
  If $\cdot \vdashi t : \M \tau_1 \otimes \M \tau_2$ then $\sem{t}$ can be factored as two morphisms $\sem{t} = f_1 \otimes f_2$, where $f_1 : I \to \M \sem{\tau_1}$ and $f_2 : I \to \M \sem{\tau_2}$.
\end{Th}

Logical relations are frequently used to prove metatheoretical properties of
type theories and programming languages. However, they are usually used in
concrete settings, i.e., for a concrete model where we can define the logical
relation explicitly. In our case, however, this approach is not enough, since we
are working with an abstract categorical semantics of \twolang. Thus, we will
leverage the categorical treatment of logical relations, called \emph{Artin
gluing}, a construction originally used in topos
theory~\citep{johnstone2007,hyland2003}.

A detailed description of this technique is beyond the scope of this paper.
However, we highlight some of the essential aspects here. We have already
introduced our class of models for \twolang. Let $\cdot \vdashi t :
\underline{\tau}$ be a well-typed program. For every concrete model $(\cat{C},
\cat{M}, \M)$, we want to show that the interpretation $\sem{t}$ in this model
satisfies some properties. At a high level, there are three steps to the gluing
argument:
\begin{enumerate}
  \item Define a category of models of \twolang, and show that every
    interpretation $\sem{\cdot}$ can be encoded as a map from the
    \emph{syntactic} model $\cat{Syn}$ to $(\cat{C}, \cat{M}, \M)$; where the
    syntactic model has types as objects and typing derivations (modulo the
    equational theory of \twolang) as morphisms. This property follows by
    showing that the syntactic model is initial.
  \item Define a triple $(\cat{Gl}(\cat{C}), \cat{M}, \widetilde{\M})$---where
    objects of the category $\cat{Gl}(\cat{C})$ are pairs $(A, X \subseteq
    \cat{C}(I, A))$, the subsets $X$ are viewed as predicates on $A$, and
    morphisms preserve these predicates---and show that this structure is a
    model of \twolang. We call this the \emph{glued} model and there is an
    obvious forgetful model morphism $(\cat{Gl}(\cat{C}), \cat{M}, \widetilde{\M})
    \to (\cat{C}, \cat{M}, \M)$.
  \item Using initiality, define a map $\bsem{\cdot}$ from the syntactic model $\cat{Syn}$ to the
    glued model. The data of this map associates every \sep-type
    $\underline{\tau}$ in \twolang to an object $(A_{\underline{\tau}},
    X_{\underline{\tau}} \subseteq \cat{C}(I, A_{\underline{\tau}}))$;
    intuitively, $A_{\underline{\tau}} \in \cat{C}$ is the interpretation of
    $\underline{\tau}$ under $\sem{\cdot}$, and the subset
    $X_{\underline{\tau}}$ encodes the logical relation at type
    $\underline{\tau}$, so this map defines a logical relation. The functor
    $\underline{\tau}$ and its codomain encode the logical relations proof.
\end{enumerate}
Finally, we can use $\bsem{\cdot}$ to map any morphism in the syntactic
category, i.e., well-typed term $\cdot \vdashi t : \underline{\tau}$, to an
element of $X_{\underline{\tau}}$. By initiality of $\cat{Syn}$, $\sem{t}$ also
is an element of $X_{\underline{\tau}}$, completing the proof by logical
relations proof. We defer the details to Appendix~\ref{app:proof}.

\section{Related Work}
\label{sec:rw}

\paragraph*{Linear logics and probabilistic programs.}
A recent line of work uses linear logic as a powerful framework to provide semantics
for probabilistic programming languages. Notably,
\citet{DBLP:journals/jacm/EhrhardPT18} show that a probabilistic version of the
coherence-space semantics for linear logic is fully abstract for probabilistic
PCF with discrete choice, and \citet{stablecones} provide a denotational
semantics inspired by linear logic for a higher-order probabilistic language
with continuous random sampling; probabilistic versions of call-by-push-value
have also been developed~\citep{tasson2019}. Linear type systems have also been
developed for probabilistic properties, like almost sure
termination~\citep{DBLP:journals/toplas/LagoG19} and differential
privacy~\citep{DBLP:conf/icfp/ReedP10,DBLP:conf/lics/AmorimGHK19}.

As we have mentioned, our categorical model for \twolang is inspired by models
of linear logic based on monoidal adjunctions, most notably Benton's
LNL~\citep{DBLP:conf/csl/Benton94}. From a programming languages perspective,
these models decompose the linear $\lambda$-calculus with exponentials in two
languages with distinct product types each: one is a Cartesian product and the
other is symmetric monoidal. The adjunction manifests itself in adding
functorial type constructor in each language, similar to our $\M$ modality.
These two-level languages are very similar to \twolang, and indeed it is
possible to show that every LNL model is a \twolang model. At the same time, the
class of models for \twolang is much broader than LNL---none of the models
presented in Section~\ref{sec:models} are LNL models.

\paragraph*{Higher-order programs and effects.}
There is a very large body of work on higher-order programs effects, which we
cannot hope to summarize here. The semantics of \onelang is an instance of
Moggi's Kleisli semantics, from his seminal work on monadic
effects~\citep{DBLP:journals/iandc/Moggi91}; the difference is that our
one-level language uses a linear type system to enforce probabilistic
independence. 

Another well-known work in this area is Call-by-Push-Value
(CBPV)~\citep{levy2001}. It is a two-level metalanguage for effects which
subsumes both call-by-value and call-by-name semantics. Each level has a
modality that takes from one level to the other one. There is a resemblance to
\twolang, but the precise relationship is unclear---none of our concrete models
are CBPV models.

Our two-level language \twolang can also be seen as an application of a novel
resource interpretation of linear logic developed by
\citet{azevedodeamorim2022sampling}, which uses an applicative modality to
guarantee that the linearity restriction is only valid for computations, not
values. Our focus is on separation and effects: we show how different sum types
for effectful computations can be naturally accommodated in this framework, we
consider a more general class of categorical models, and we prove a soundness
theorem ensuring separation for effectful computations.

\paragraph{Bunched type systems.}
Our focus on sharing and separation is similar to the motivation of another
substructural logic, called the logic of bunched implicates
(BI)~\citep{DBLP:journals/bsl/OHearnP99}.  Like our system, BI features two
conjunctions modeling separation of resources, and sharing of resources. Like in
\onelang, these conjunctions in BI belong to the same language. Unlike our work,
BI also features two implications, one for each conjunction. The leading
application of BI is in separations logic for concurrent and heap-manipulating
programs~\citep{DBLP:conf/csl/OHearnRY01,DBLP:conf/iwmm/OHearn07}, where pre-
and post-conditions are drawn from BI.

Most applications of BI use a truth-functional, Kripke-style
semantics~\citep{DBLP:journals/tcs/PymOY04}. By considering the proof-theoretic
models of BI, \citet{DBLP:journals/jfp/OHearn03} developed a bunched type system
for a higher-order language. Its categorical semantics is given by a
\emph{doubly closed category}: a Cartesian closed category with a separate
symmetric monoidal closed structure. While \citet{DBLP:journals/jfp/OHearn03}
showed different models of this language for reasoning about sharing and
separation in heaps, few other concrete models are known. It is not clear how to
incorporate effects into the bunched type system; in contrast, our models can
reason about a wide class of monadic effects. 

There are natural connections to both of our languages. Our language \onelang
resembles O'Hearn's system, with two differences. First, \onelang only has a
multiplicative arrow, not an additive arrow---as we described in
\Cref{sec:langindep}, it is not clear how to support an additive arrow in
\onelang without breaking our primary soundness property. Second, contexts in
\onelang are flat lists, not tree-shaped bunches; it would be interesting to use
bunched contexts to represent more complex dependency relations.

Our stratified language \twolang is also similar to O'Hearn's system. Though our
categorical model only has a single multiplicative arrow, in the \sep-layer,
many---but not all---of our concrete models also support an additive arrow, in
the \shr-layer.  Furthermore, by assuming a single category, instead of two
categories as in our approach, in BI it is possible to layer the connectives
$\times$ and $\otimes$ to create intricate dependency structures. In contrast
our two-layer language only allows to create dependencies of the form $\M(\tau
\times \cdots \times \tau) \otimes \cdots \otimes \M(\tau \times \cdots \times
\tau)$.  At the same time, it is not clear how the two sum types in \twolang
would function in a bunched type system.

\paragraph{Probabilistic independence in higher-order languages.}
There are a few probabilistic functional languages with type systems that model
probabilistic independence. Probably the most sophisticated example is due
to~\citet{darais2019}, who propose a type system combining linearity,
information-flow control, and probability regions for a probabilistic functional
language. \citet{darais2019} show how to use their system to implement and
verify security properties for implementations of oblivious RAM (ORAM). Our work
aims to be a core calculus capturing independence, with a clean categorical
model.

\citet{DBLP:journals/toplas/VesgaRG21} present a probabilistic functional
language embedded in Haskell, aiming to verify accuracy properties of programs
from differential privacy.  Their system uses a taint-based analysis to
establish independence, which is required to soundly apply concentration bounds,
like the Chernoff bound. Unlike our work, \citet{DBLP:journals/toplas/VesgaRG21}
do not formalize their independence property in a core calculus.

\paragraph{Probabilistic separation logics.}
A recent line of work develops separation logics for first-order, imperative
probabilistic programs, using formulas from the logic of bunched implications to
represent pre- and post-conditions. Systems can reason about probabilistic
independence~\citep{barthe2019}, but also refinements like conditional
independence~\citep{bao2021}, and negative association~\citep{bao2022}. These
systems leverage different Kripke-style models for the logical assertions; it is
unclear how these ideas can be adapted to a type system or a higher-order
language. There are also quantitative probabilistic separation
logics~\citep{DBLP:journals/pacmpl/BatzKKMN19,DBLP:conf/esop/BatzFJKKMN22}.

\section{Conclusion and Future Directions}
\label{sec:conc}

We have presented two linear, higher-order languages with types that can capture
probabilistic independence, and other notions of separation in effectful
programs. We see several natural directions for further investigation.

\paragraph{Other variants of independence.}
In some sense, probabilistic independence is a trivial version of dependence: it
captures the case where there is no dependence whatsoever between two random
quantities. Researchers in statistics and AI have considered other notions that
model more refined dependency relations, such as conditional independence,
positive association, and negative dependence
(e.g.,~\citep{DBLP:journals/rsa/DubhashiR98}). Some of these notions have been
extended to other models besides probability; for instance,
\citet{DBLP:conf/ecai/PearlP86} develop a theory of \emph{graphoids} to
axiomatize properties of conditional independence. It would be interesting to
see whether any of these notions can be captured in a type system.

\paragraph{Bunched type systems for independence.}
Our work bears many similarity to work on bunched logics; most notably, bunched
logics feature an additive and a multiplicative conjunction. While bunched
logics have found strong applications in Hoare-style logics, the only bunched
type system we are aware of is due to \citet{DBLP:journals/jfp/OHearn03}. This
language features a single layer with two product types and also two function
types, and the typing contexts are tree-shaped bunches, rather than flat lists.
Developing a probabilistic model for a language with a richer context structure
would be an interesting avenue for future work.

\paragraph{Non-commutative effects.}
Our concrete models encompass many kinds of monadic effects, but we only support effects
modeled by commutative monads. Many common effects are modeled by
non-commutative monads, e.g., the global state monad. It may be possible to
extend our language to handle non-commutative effects, but we would likely need
to generalize our model and consider non-commutative logics.

\paragraph{Towards a general theory of separation for effects.}
We have seen how in the presence of effects, constructs like sums and products
come in two flavors, which we have interpreted as sharing and separate. Notions
of sharing and separation have long been studied in programming languages and
logic, notably leading to separation logics. We believe that there should be a
broader theory of separation (and sharing) for effectful programs, which still
remains to be developed.

\bibliography{mybib}

\appendix

\section{Categorical Soundness Proof: Details}
\label{app:proof}

\subsection{Category of Models}

A model for \twolang is given by a CD category $\cat{M}$ with coproducts, a SMCC
$\cat{C}$ with coproducts and a lax monoidal functor $\M: \cat{M} \to \cat{C}$.
A morphism between two models $(\cat{M}_1, \cat{C}_1, \M_1)$ and $(\cat{M}_2,
\cat{C}_2, \M_2)$ is a pair of functors $(F : \cat{M}_1 \to \cat{M}_2, G :
\cat{C}_1 \to \cat{C}_2)$ that preserves the logical connectives. By defining
morphism composition component-wise and the pair $(id_{\cat{C}}, id_{\cat{M}})$
as the identity morphism, this structure constitutes a category which we call
$\cat{Mod}$. 

In categorical treatments of type theories it is important to show that
the equational theory is a sound approximation of the categorical semantics.
In the case of \twolang, since the language does not use any fancy type
theoretic constructions, the soundness property is straightforward to prove
by induction of the typing derivations.

\begin{Th}
  \label{th:eqsound}
  Let $(\cat{C}, \cat{M}, \M)$ be a \twolang model. If $\Gamma \vdashni M \equiv N : \tau$ then $\sem{M} = \sem{N}$
  and if $\Gamma \vdashi t \equiv u : \tau$ then $\sem{t} = \sem{u}$.
\end{Th}

The main subtlety is that we have to be a bit more precise in the presentation of the equational
theory for the \sep language. Note that the $\mathsf{sample}$ construct can sample simultaneously
from any number of distributions, while applicative functors only provide
a binary sampling operator. Formally this is resolved by restricting $\mathsf{sample}$
to two arguments and add the following rules to the equational theory.
\begin{mathpar}
  \inferrule[]{\Gamma_i \vdashi t_i : \M \tau_i \\ i \in \{ 1, 2, 3\}}{\Gamma_1, \Gamma_2, \Gamma_3 \vdashi \sample{t_1, (\sample{t_2, t_3}{x_2, x_3}{(x_2, x_3))}}{x_1, y}{(x_1, \pi_1\, y, \pi_2 \, y)} \equiv\\ \sample{(\sample{t_1, t_2}{x_1, x_2}{(x_1, x_2)), t_3}}{y, x_3}{( \pi_1\, y, \pi_2 \, y, x_3)} : \M(\tau_1 \times \tau_2 \times \tau_3)}
  \\
  \inferrule[]{\Gamma \vdashi t : \M \tau}{\Gamma \vdashi \sample{t, (\sample{\_}{\_}{()})}{x, y}{x} \equiv t : \M \tau}
\\
  \inferrule[]{\Gamma \vdashi t : \M \tau}{\Gamma \vdashi \sample{(\sample{\_}{\_}{()}), t}{x, y}{y} \equiv t : \M \tau}
\end{mathpar}

Note that even though the rule looks intimidating, it is basically the lax monoidal commutativity diagram in syntax form,
which says that the sample operation is associative and, as a consequence, there is a unique
way of defining the $n$-ary operation $\sample{t_1, \dots t_n}{x_1,\dots, x_n}{M}$, for $n \geq 2$.

An important \twolang model is the syntactic object $\cat{Syn}$, which is a
triple $(\cat{Syn}_{lin}, \cat{Syn}_{CD}, \M)$, where $\cat{Syn}_{CD}$ is the
syntactic category of CD categories with coproducts while $\cat{Syn}_{lin}$ is
the syntactic category of symmetric monoidal closed categories with coproducts
and an applicative modality and $\M$ is the type constructor for the modality.
Concretely each of these categories have types as objects and morphisms are
programs with one free variables modulo the equational theories presented in
\Cref{fig:twoeq}. It follows by a simple inspection that $\cat{Syn}$ is a
\twolang model.
\begin{Th}
  $\cat{Syn}$ is a \twolang model.
\end{Th}

\begin{Th}
\label{th:initial}
  $\cat{Syn}$ is the initial object of $\cat{Mod}$.
\end{Th}
\begin{proof}
  Let $(\cat{C}, \cat{M}, \M)$ be a model. It is possible to construct a
  morphism $\sem{\cdot} : \cat{Syn} \to (\cat{C}, \cat{M}, \M)$ by defining two functors $\sem{\cdot}_1 : \cat{Syn}_{lin} \to \cat{C}$ and
  $\sem{\cdot}_2 : \cat{Syn}_{CD} \to \cat{M}$. Since $\cat{Syn}_{lin}$ and $\cat{Syn}_{CD}$ are freely generated, the action of the functors
  on objects is characterized by a simple induction on the types.
  The action on morphisms is defined by induction on the typing derivation using
  \Cref{fig:twosem}.

  The proof that this function is well-defined follows from Theorem~\ref{th:eqsound}. Uniqueness follows by assuming the existence of two semantics 
  and showing, by induction on the typing derivation, that they are equal.
\end{proof}

\subsection{Glued category}

We construct the logical relations category by using a comma category. Formally,
a comma category along functors $F :  \cat{C_1} \to \cat{D}$ and $G : \cat{C_2}
\to \cat{D}$ has triples $(A, X, h)$ as objects, where $A$ is an $\cat{C_1}$
object, $X$ is an $\cat{C_2}$ objects and $h : FA \to G X$, and its morphisms
$(A, X, h) \to (A', X', h')$ are pairs $f : A \to A'$ and $g : X \to X'$  making
certain diagrams commute. In Computer Science applications of gluing, it is
usually assumed that $F$ is the identity functor and $\cat{D} = \cat{Set}$.
Furthermore, to simplify matters, sometimes it is also assumed that we work with
full subcategories of the glued category, for instance we can assume that we
only want objects such that $A \to G B$ is an injection, effectively
representing a subset of $G B$.

Therefore, in the setting we are interested in a glued category along a functor
$G : \cat{C} \to \cat{Set}$ has pairs $(A, X\subseteq G(A))$ as objects and its
morphisms $(A, X) \to (B, Y)$ is a $\cat{C}$ morphism $f : A \to B$ such that
$G(f)(X) \subseteq Y$. Note that this condition can be seen as a more abstract
way of phrasing the usual logical relations interpretation of arrow types:
mapping related things to related things. At an intuitive level we want to use
the functor $G$ to map types to predicates satisfied by its inhabitants.

Now, we are ready to define the glued category and show that it constitutes a
model for the language. Given a triple $(\cat{M}, \cat{C}, \M)$ we define the
triple $(\cat{M}, \cat{Gl(C)}, \widetilde{\M})$, where the objects of
$\cat{Gl(C)}$ are pairs $(A \in \cat{C}, X \subseteq \cat{C}(I, A))$ and the
morphisms are $\cat{C}$ morphisms that preserve $X$, i.e. we are gluing $\cat{C}$
along the global sections functor $\cat{C}(I, -)$. The functor $\M : \cat{M}
\to \cat{C}$ is lifted to a functor $\widetilde{\M} : \cat{C} \to \cat{Gl(C)}$
by mapping objects $X$ to $(\M\, X, \cat{C}(I, \M\, X))$ and by mapping
morphisms $f$ to $\M \, f$.\footnote{%
  Note that its predicate set is every $\cat{C}$ morphism $I \to \M\, X$,
  similar to how ground types are interpreted in usual logical relations proofs.}
Now we have to show that the triple is indeed a model of our language.

Something that simplifies our proofs is that morphisms in $\cat{Gl(C)}$ are
simply morphisms in $\cat{C}$ with extra structure and composition is kept the
same. Therefore,  once we establish that a $\cat{C}$ morphism is also a
$\cat{Gl(C)}$ morphism all we have to do in order to show that a certain
$\cat{Gl(C)}$ diagram commutes is to show that the respective $\cat{C}$ diagram
commutes.

\begin{Th}
  $\cat{Gl(C)}$ is a SMCC with coproducts and with a natural transformation $del$.
\end{Th}
\begin{proof}
  Let $(A, X)$ and $(B, Y)$ be $\cat{Gl(C)}$ objects, we define $(A, X) \otimes
  (B, Y) = (A \otimes B, \set{f : I \to A \otimes B}{ f = f_A \otimes f_B, f_A
  \in X, f_B \in Y})$. The monoidal unit is given by $(I, \cat{C}(I, I))$ and the
  natural transformation $del$ is the same one as the one in $\cat{C}$, which
  is a morphism in $\cat{Gl(C)}$ because $X_I = \cat{C}(I, I)$.

  Let $(A, X)$ and $(B, Y)$ be $\cat{Gl(C)}$ objects, we define $(A, X) \lto (B,
  Y) = (A \lto B, \set{f : I \to (A \lto B)}{\forall f_A \in X_A, \epsilon_B
  \circ (f_A \otimes f) \in X_B}$, where $\epsilon_B : (A \lto B) \otimes A \to
  B$ is the counit of the monoidal closed adjunction.

  To show $A \otimes ( - ) \dashv A \lto (-)$ we can use the (co)unit
  characterization of adjunctions, which corresponds to the existence of two
  natural transformations $\epsilon_B : A \otimes (A \lto B) \to B$ and $\eta_B
  : B \to A \lto (A \otimes B)$ such that $1_{A \otimes -} = \epsilon (A \otimes
  -) \circ (A \otimes -) \eta$ and $1_{A \lto -} = (A \lto -) \epsilon \circ
  \eta (A \lto -)$, where $1_{F}$ is the identity natural transformation between
  $F$ and itself. By choosing these natural transformations to be the same as in
  $\cat{C}$, since the adjoint equations hold for them by definition, all we
  have to do is show that they are also $\cat{Gl(C)}$ morphisms, which follows
  by unfolding the definitions.

  Finally, we can show that $\cat{Gl(C)}$ has coproducts. Let $(A_1, X_1)$ and
  $(A_2, X_2)$ be $\cat{Gl(C)}$ objects, we define $(A_1, X_1) \oplus (A_2, X_2) =
  (A_1 \oplus A_2, \set{\mathsf{in_i}\, f_i}{f_i \in X_i})$.  To show that it
  satisfies the universal property of sum types. Let $f_1 : A_1 \to B$ and $f_2 :
  A_2 \to B$ be $\cat{Gl(C)}$ morphisms. Consider the $\cat{C}$ morphism $[f_1,
  f_2]$. We want to show that this morphism is also a $\cat{Gl(C)}$ morphism.
  Consider $g \in X_{A_1 \oplus A_2}$ which, by assumption, $g = \mathsf{in}_1
  g_1$ or $g = \mathsf{in}_2$. By case analysis and the facts $f_i \circ g_i \in
  Y$ and $[f_1, f_2] \circ \mathsf{in}_i g_i = f_i \circ g_i$ we can conclude that
  $[f_1, f_2]$ is indeed a $\cat{Gl(C)}$ morphism.
\end{proof}

These constructions are known in the categorical logic literature \citep{hyland2003},
but since it is simple enough we think that it is helpful to also present it here.
Since every construction so far uses the same objects as the ones in $\cat{C}$,
it is possible to show that the forgetful functor $U: \cat{Gl(C)} \to \cat{C}$
preserves every type constructor and is a $\cat{Mod}$ morphism. Next, we have to
show that $\widetilde{\M}$ is lax monoidal which follows from the fact that
$\mu$ and $\epsilon$ preserve the plot sets, by a simple unfolding of the
definitions. We can now easily conclude that the lax monoidality diagrams
commute, since composition is the same and $\M$ is lax monoidal.

Thus, the glued category is a model.

\begin{Th}
  The triple $(\cat{M}, \cat{Gl(C)}, \widetilde{\M})$ is a $\cat{Mod}$ object.
\end{Th}

There is a forgetful map from the glued model to the original model.

\begin{Th}
  There is a $\cat{Mod}$ morphism $U : (\cat{M}, \cat{Gl(C)}, \widetilde{\M})
  \to (\cat{M}, \cat{C}, \M)$.
\end{Th}

Finally, by initiality of $\cat{Syn}$, we can prove

\begin{Th}
  There is a $\cat{Mod}$ morphism $\bsem{\cdot} : \cat{Syn} \to (\cat{M}, \cat{Gl(C)}, \widetilde{\M})$.
\end{Th}

With this map in hand, we may now construct a functor $U \circ \bsem{\cdot} :
\cat{Syn} \to (\cat{M}, \cat{C}, \M)$ which, by initiality of $\cat{Syn}$, is
equal to the functor $\sem{\cdot}$, as illustrated by Figure~\ref{fig:gluing}.
\begin{figure}
\[\begin{tikzcd}
	\cat{Syn} \\
	{(\cat{M}, \cat{Gl(C), \widetilde\M})} & (\cat{M}, \cat{C}, \M)
	\arrow["{\bsem{\cdot}}"', from=1-1, to=2-1]
	\arrow["{\sem{\cdot}}", from=1-1, to=2-2]
	\arrow["U"', from=2-1, to=2-2]
\end{tikzcd}\]
\hrule
\caption{The essence of the soundness proof}
\label{fig:gluing}
\end{figure}

\subsection{General Soundness Theorem}

\begin{Th} \label{th:gensound}
  If $\cdot \vdashi t : \underline{\tau}$, then $\sem{t} \in
  X_{\underline{\tau}}$.
\end{Th}
\begin{proof}
  We know that $\sem{\cdot} = U \circ \bsem\cdot$ and that $\bsem{t}$ is a
  $\cat{Gl(C)}$ morphism. As such we have that $\sem{t} = \bsem{t} = \bsem{t}
  \circ id_I \in X_{\underline{\tau}}$, since, by definition, $id_I \in X_I$.
\end{proof}

Theorem~\ref{th:soundness} follows immediately, as a corollary.

\begin{Cor}
  If $\cdot \vdashi t : \M \tau_1 \otimes \M \tau_2$ then $\sem{t}$ can be factored as two morphisms $\sem{t} = f_1 \otimes f_2$, where $f_1 : I \to \M \sem{\tau_1}$ and $f_2 : I \to \M \sem{\tau_2}$.
\end{Cor}
\begin{proof}
  By Theorem~\ref{th:gensound}, if $\cdot \vdashi t : \M\tau_1 \otimes \M \tau_2$, then $\sem{t} \in X_{\M \tau_1 \otimes \M \tau_2}$ which, by unfolding the definitions, means that there exists $f_1 : I \to \M \sem{\tau_1}$ and $f_2 : I \to \M \sem{\tau_2}$ such that $\sem t = f_1 \otimes f_2$.
\end{proof}

\section{Measurable sets and Markov Kernels}
\label{app:meas}

A measurable space combines a set with a collection of subsets, describing the
subsets that can be assigned a well-defined measure or probability. 

\begin{Def}
  Given a set $X$, a $\sigma$-\emph{algebra} $\Sigma_X \subseteq \mathcal{P}(X)$
  is a set of subsets such that (i) $X \in \Sigma_X$, and (ii) $\Sigma_X$ is
  closed complementation and countable union. A \emph{measurable space} is a
  pair $(X, \Sigma_X)$, where $X$ is a set and $\Sigma_X$ is a $\sigma$-algebra.

  A \emph{measurable function} between measurable spaces $(X, \Sigma_X)$ and
  $(Y, \Sigma_Y)$ is a function $f : X \to Y$ such that for every
  $A\in\Sigma_Y$, $f^{-1}(A) \in \Sigma_X$, where $f^{-1}$ is the inverse image
  function. Measurable spaces and measurable functions form a category
  $\cat{Meas}$.
\end{Def}

\begin{Def}
    Standard Borel spaces $(X, \Sigma_X)$ are spaces such that $X$ can be equipped
    with a metric such that $X$ is, as a metric space, complete and separable and
    $\Sigma_X$ is the $\sigma$-algebra generated by the metric.
\end{Def}

\begin{Ex}
    For every $n \in \nat$, $\R^n$ with its standard $\sigma$-algebra is a standard Borel space.
\end{Ex}

\begin{Def}
  A \emph{probability measure} is a function $\mu_X : \Sigma_X \to [0, 1]$ such
  that: (i) $\mu(\emptyset) = 0$, (ii) $\mu(X) = 1$, and $\mu(\uplus A_i) =
  \sum_i \mu(A_i)$.
\end{Def}

\begin{Def}
  A \emph{Markov kernel} between measurable spaces $(X, \Sigma_X)$ and $(Y,
  \Sigma_Y)$ is a function $f : X \times \Sigma_Y \to [0,1]$ such that:
  \begin{itemize}
      \item For every $x \in X$, $f(x,-)$ is a probability distribution.
      \item For every $B \in \Sigma_Y$, $f(-, B)$ is a measurable function.
  \end{itemize}
\end{Def}

Markov kernels $f : X \times \Sigma_Y \to [0,1]$ and $g : Y \times \Sigma_Z \to [0,1]$
can be composed with the following formula
\[(g \circ f)(x, C) = \int g(-, C)d f(x,-)\]
The Dirac kernel $\delta(a, A) = 1$ if $a \in A$ and $0$ otherwise 
is the unit for the composition defined above that this structure
can be organized into a category $\cat{BorelStoch}$ with standard
Borel spaces as objects and Markov kernels as morphisms.

\paragraph{Marginals and probabilistic independence.}
We will need some constructions on distributions and measures over products.

\begin{Def}
  Given a distribution $\mu$ over $X \times Y$, its \emph{marginal} $\mu_X$ is
  the distribution over $X$ defined by $\mu_X(A) = \int_Y d \mu(A, -)$.
  Intuitively, this is the distribution obtained by sampling a pair from $\mu$
  and projecting to its first component. The other marginal $\mu_Y$ is defined
  similarly.
\end{Def}

\begin{Def}
  A probability measure $\mu$ over $A \times B$ is probabilistically
  \emph{independent} if it is a product of its marginals $\mu_A$ and $\mu_B$,
  i.e., $\mu(X, Y) = \mu_A(X) \cdot \mu_B(Y)$, $X \in \Sigma_A$ and $Y \in
  \Sigma_B$.
\end{Def}

\end{document}